\documentclass[final,1p,times,twocolumn]{elsarticle}
\usepackage{amssymb,amsthm,amssymb,mathrsfs,amsmath}
\usepackage{color}
\newtheorem{thm}{Theorem}[section]
\newtheorem{prop}{Proposition}[section]
\newtheorem{lem}{Lemma}[section]

\newtheorem{cor}{Corollary}[section]

\newtheorem{defn}{Definition}[section]

\newtheorem{rmk}{\it Remark}[section]

\newcommand{\Log}{\mathrm{Log}}
\newcommand{\R}{\mathrm{Re}}
\newcommand{\I}{\mathrm{Im}}
\newcommand{\sign}{\mathrm{sign}}
\newcommand{\const}{\mathrm{const.}}
\newcommand{\cc}{\mathrm{c.c.}}

\newcommand{\Res}{\mathrm{Res}}

\def\Xint#1{\mathchoice
{\XXint\displaystyle\textstyle{#1}}%
{\XXint\textstyle\scriptstyle{#1}}%
{\XXint\scriptstyle\scriptscriptstyle{#1}}%
{\XXint\scriptscriptstyle\scriptscriptstyle{#1}}%
\!\int}
\def\XXint#1#2#3{{\setbox0=\hbox{$#1{#2#3}{\int}$ }
\vcenter{\hbox{$#2#3$ }}\kern-.6\wd0}}

\def\dashint{\Xint-}

\journal{Physica D}

\begin{document}

\begin{frontmatter}

%% Title, authors and addresses

%% use the tnoteref command within \title for footnotes;
%% use the tnotetext command for theassociated footnote;
%% use the fnref command within \author or \address for footnotes;
%% use the fntext command for theassociated footnote;
%% use the corref command within \author for corresponding author footnotes;
%% use the cortext command for theassociated footnote;
%% use the ead command for the email address,
%% and the form \ead[url] for the home page:
%% \title{Title\tnoteref{label1}}
%% \tnotetext[label1]{}
%% \author{Name\corref{cor1}\fnref{label2}}
%% \ead{email address}
%% \ead[url]{home page}
%% \fntext[label2]{}
%% \cortext[cor1]{}
%% \address{Address\fnref{label3}}
%% \fntext[label3]{}

\title{The scattering transform for the Benjamin-Ono equation in the small-dispersion limit\tnoteref{toWhitham}}
\tnotetext[toWhitham]{In memory of G. B. Whitham.}

%% use optional labels to link authors explicitly to addresses:
%% \author[label1,label2]{}
%% \address[label1]{}
%% \address[label2]{}

\author{Peter D. Miller}\ead{millerpd@umich.edu}
\author{Alfredo N. Wetzel}\ead{wreagan@umich.edu}
\address{Department of Mathematics, University of Michigan, East Hall, 530 Church St., Ann Arbor, MI 48109}
%\cortext[cor1]{Corresponding author.}

\begin{abstract}
Using exact formulae for the scattering data of the Benjamin-Ono equation valid for general rational potentials recently obtained in \cite{MillerWetzel2015}, we rigorously analyze the scattering data in the small-dispersion limit.
In particular, we deduce precise asymptotic formulae for the reflection coefficient, the location of the eigenvalues and their density, and the asymptotic dependence of the phase constant (associated with each eigenvalue) on the eigenvalue itself.  Our results give direct confirmation of conjectures in the literature that have been partly justified by means of inverse scattering, and they also provide new details not previously reported in the literature.
\end{abstract}


\end{frontmatter}

\section{Introduction}\label{introduc}

The Benjamin-Ono (BO) equation
\begin{equation}
\frac{\partial u}{\partial t} + 2 u \frac{\partial u}{\partial x} + \epsilon \mathcal{H}\left[ \frac{\partial^2 u}{\partial x^2} \right] =0, \quad -\infty < x < \infty, \quad t >0 \label{BOeqn}
\end{equation}
describes the weakly nonlinear evolution of one-dimensional internal gravity waves in a stratified fluid \cite{Benjamin67, DavisAcrivos67, Ono75}, where $u$ corresponds to the wave profile, $\epsilon > 0$ is a measure of the effect of dispersion, and the operator $\mathcal{H}$ denotes the Hilbert transform defined by the Cauchy principal value integral
\begin{equation}
\mathcal{H}[u](x,t) := \frac{1}{\pi} \dashint_{-\infty}^{\infty} \frac{u(\xi,t)}{\xi-x}\; d\xi.
\end{equation}
The BO equation  \eqref{BOeqn} has been used to model internal waves in deep water \cite{ChoiCamassa1996}, the atmospheric roll cloud wave-train known as the morning glory \cite{PorterSmyth2002}, and nonlinear Rossby waves in shear flow \cite{Ono1981}, to name a few examples.
Beyond fluid dynamics, equation \eqref{BOeqn} has been observed to model the spectral dynamics of incoherent shocks in nonlinear optics \cite{GarnierXuTrilloPicozzi2013}.

In the small-dispersion limit ($\epsilon \to 0$), numerical experiments indicate that, for $t$ sufficiently small (independent of $\epsilon$), the solution of the BO equation with smooth $\epsilon$-independent initial data $u_0$ is well approximated by the solution of the inviscid Burgers equation (Hopf equation) obtained by setting $\epsilon=0$ in \eqref{BOeqn}. 
At the advent of a shock in the dispersionless model, the solution of \eqref{BOeqn} is then  regularized with the formation of a dispersive shock wave (DSW).
This DSW has the well-known structure of an $O(1)$-amplitude modulated periodic traveling wave with wavelength $O(\epsilon)$. 
In the DSW region the solution of the BO equation may be formally approximated using Whitham modulation theory. 
Surprisingly, and unlike the case of the Korteweg-de Vries (KdV) equation, the modulation equations for the BO equation are fully uncoupled \cite{DobrokhotovKrichever1991}, consisting of several independent copies of the inviscid Burgers equation.
A formalism for matching the Whitham modulation approximation for the DSW onto inviscid Burgers solutions in the domain exterior to the DSW was developed and applied by Matsuno \cite{Matsuno1998a, Matsuno1998b} and Jorge, Minzoni, and Smyth \cite{JorgeMinzoniSmyth1999} to analyze the Cauchy problem for \eqref{BOeqn}.
Partial confirmation of these results was given by Miller and Xu \cite{MillerXu2011}, who rigorously computed the weak limit of (essentially, modulo an approximation of the scattering data) the solution of the Cauchy problem for \eqref{BOeqn} for a class of positive initial data using inverse-scattering transform (IST) techniques, by developing an analogue for the BO equation of a method first invented for KdV by Lax and Levermore \cite{LaxLevermore1983}. 
Avoiding the (a priori unjustified) approximation of the scattering data requires its careful analysis in the small-dispersion limit by direct means.

One application of small-dispersion theory in nonlinear waves involves the investigation of universality, for example determining properties of general solutions near the onset of the DSW (the point of gradient catastrophe for the inviscid Burgers equation) that are independent of initial data. 
It has been postulated by Dubrovin \cite{Dubrovin2006} that the solutions of a wide class of dispersive Hamiltonian perturbations of the inviscid Burgers equation are universally modeled by a particular solution of an integrable fourth order ODE (Painlev\'e I$_2$) near the gradient catastrophe point.
This result was subsequently proven rigorously for the KdV equation by Claeys and Grava \cite{ClaeysGrava2008} using Riemann-Hilbert analysis in an appropriate double scaling limit.
Interestingly, the BO equation does not fall within the Dubrovin universality class of equations.  However,  Maseoro, Raimondo, and Antunes \cite{MasoeroRaimondoAntunes2015} have recently generalized Dubrovin's method to the BO and other nonlocal equations.  Based on this analysis, in \cite{MasoeroRaimondoAntunes2015} it is postulated that the solution of the BO Cauchy problem should be universally modeled near the point of gradient catastrophe by a particular solution of a certain nonlocal analogue of Painlev\'e I$_2$.
Rigorous verification of this formal result using the IST would again require, among other tools related to asymptotic analysis of nonlocal Riemann-Hilbert problems yet to be developed, the accurate asymptotic characterization of the scattering data for the BO equation.

In this paper we use recently obtained exact formulae \cite{MillerWetzel2015} to study  the asymptotic behavior of the scattering data for the BO equation \eqref{BOeqn} in the limit $\epsilon \to 0$. 
These formulae hold for rational initial data with simple poles, and for convenience we adopt here the further condition of a single local extremum (Definition~\ref{def:rKS}). 
The formulae for the scattering data are sufficiently explicit to admit rigorous asymptotic analysis by classical methods such as steepest descent (for integrals), thus justifying and extending previously reported formal asymptotic results.
We confirm directly the celebrated formulae of Matsuno for the asymptotic density of the eigenvalues \cite{Matsuno81} (Corollary~\ref{cor-Matsuno-density}) and the magnitude of the reflection coefficient \cite{Matsuno82} (Corollary~\ref{beta_mag_asymp}), the formula postulated in Miller and Xu \cite{MillerXu2011} for the asymptotic values of the phase constants (Theorem~\ref{theorem-gamma-limit}), and Xu's conjecture on the phase of the reflection coefficient \cite{Xuthesis} with some modifications (Theorem~\ref{Reflect_asymptot}).  We also obtain more precise information about the discrete spectrum than had even been conjectured before (Theorem~\ref{theorem-uniform} and Corollary~\ref{cor_asymp_eig_loc}).

\section{Direct Scattering for Rational Initial Data}\label{sec:direct-scattering-rational}

The direct scattering problem for the BO equation can be thought of as the problem of spectral analysis of the operator $\mathcal{L}:=-i\epsilon\partial_x + \mathcal{C}^+u_0\mathcal{C}^+$ acting on the Hardy subspace $H^+(\mathbb{R})$ of $L^2(\mathbb{R})$ consisting of functions analytic in the upper-half $x$-plane.  Here, the potential $u_0:\mathbb{R}\to\mathbb{R}$ is the initial condition $u(x,0)=u_0(x)$ for \eqref{BOeqn}, and the operator $\mathcal{C}^+:L^2(\mathbb{R})\to H^+(\mathbb{R})$ is the self-adjoint orthogonal projection onto the Hardy space.
As first noted in \cite{KodamaAblowitzSatsuma82} and generalized in \cite{MillerWetzel2015}, the direct scattering problem for the BO equation can be effectively solved for rational potentials of the form 
\begin{equation}\label{u0_def_frac}
u(x,0) = u_0(x) = \sum_{p=1}^{P} \frac{c_p}{x-z_p}+\cc
\;\; \text{satisfying} \;\;
u_0 \in L^{1}(\mathbb{R}),
\end{equation}
where $c_p \neq 0$ and the poles $\{z_p\}_{p=1}^{P}$ are points with positive imaginary parts and distinct real parts increasing with $p$.  The condition $u_0\in L^1(\mathbb{R})$ is equivalent to $\R\{c_1+\cdots + c_P\}=0$, in which case $M:=\int_{-\infty}^\infty u_0(x)\,dx=2\pi i(c_1+\cdots + c_P)\in\mathbb{R}$.  For potentials of the form \eqref{u0_def_frac}, the problem of recovering the scattering data is essentially reduced to the problem of studying the linear system
\begin{equation}\label{lin_sys_Abar}
\mathbf{A}(\lambda) \mathbf{v}(\lambda) = \mathbf{b}(\lambda)
\end{equation}
for an unknown vector $\mathbf{v}(\lambda) \in \mathbb{C}^{P}$, where $\mathbf{A}(\lambda) \in \mathbb{C}^{P \times P}$ and $\mathbf{b}(\lambda) \in \mathbb{C}^{P}$ have components
\begin{equation}\label{lin_sys_coeff}
A_{mp}(\lambda) := \int_{C_m} \frac{e^{-ih(z;\lambda)/\epsilon }}{z-z_p} \, dz
\; \;\text{and} \;\;
b_m(\lambda) := -  \lambda \int_{C_m} e^{-ih(z;\lambda)/\epsilon } \, dz
\end{equation}
in which the exponents involve the function
\begin{equation}\label{def_h}
h(x;\lambda) := \lambda x + f(x),
\end{equation}
and $f$ denotes a certain anti-derivative of $u_0$:
\begin{equation} \label{f_eqn}
f(x) := \sum_{p=1}^{P}  c_p \left( \Log \left( i (x-z_p)\right)  +\frac{\pi i}{2} \right)+ \cc,\;\;
f'(x)=u_0(x),
\end{equation}
where $\Log(\cdot)$ is the principal branch ($|\I\{\Log(\cdot)\}|<\pi$).  The contours $C_1,\dots,C_P$
lie in the domain of analyticity of $f$ and are chosen as follows when $\lambda<0$.  If $ic_m/\epsilon$ is not a negative integer, then $C_m=U_m^<$, a contour beginning at $i\infty$ to the left of the line $\R\{z\}=\R\{z_1\}$ and terminating at $i\infty$ to the right of the line $\R\{z\}=\R\{z_m\}$ and (if $m<P$) to the left of the line $\R\{z\}=\R\{z_{m+1}\}$.  If $ic_m/\epsilon$ is a negative integer, then $C_m=\ell_0(z_m)$ is a contour beginning at $i\infty$ to the left of the line $\R\{z\}=\R\{z_1\}$ and terminating at $z_m$.  See Figure~\ref{fig:U-combo-PhysicaD}.
\begin{figure}[h!]
\begin{center}
\includegraphics[scale=0.8]{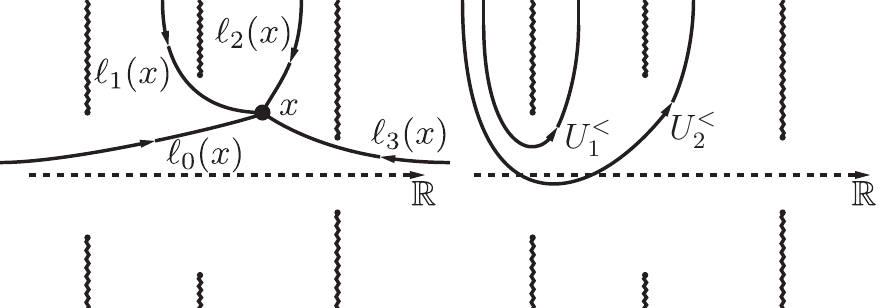}
\caption{The domain of analyticity of the function $f$ omits from $\mathbb{C}$ vertical branch cuts emanating from each of the poles of $u_0$ in the direction away from the real axis, illustrated here for a rational potential $u_0$ with $P=3$.  Left:  contours $\ell_m(x)$.  Right:  contours $U_m^<$.} \label{fig:U-combo-PhysicaD}
\end{center}
\end{figure}
The matrix $\mathbf{A}(\lambda)$ and vector $\mathbf{b}(\lambda)$ can be analytically continued from $\lambda<0$ to the domain $\mathbb{C}\setminus\mathbb{R}^+$, and while the formulae \eqref{lin_sys_coeff} remain valid, due to the dominant term $\lambda x$ in the function $h$, the contours $C_1,\dots,C_P$ must be suitably rotated (continuing $f$ through its branch cuts) to maintain absolute convergence of the integrals and allow the continuation.  

The discrete eigenvalues of the BO spectral problem for the rational potential $u_0$ are the roots of the \emph{Evans function}
\begin{equation}
\label{Evans_function}
D(\lambda):=\det(\mathbf{A}(\lambda)),\quad\lambda\in\mathbb{C}\setminus\mathbb{R}^+.  
\end{equation}
As $\mathcal{L}$ is self-adjoint, the eigenvalues are real numbers.  Moreover, since $-i\epsilon\partial_x$ is positive semi-definite on $H^+(\mathbb{R})$, it is easy to see that the eigenvalues $\lambda$ satisfy 
\begin{equation}
\lambda\ge -\sup_{x\in\mathbb{R}}\{u_0(x)\},\quad\text{$\lambda$ an eigenvalue}.  
\label{eq:eigenvalue-bound}
\end{equation}
Associated with each eigenvalue $\lambda=\lambda_j<0$ there is a unique normalized (although not in the $L^2(\mathbb{R})$-sense) eigenfunction $\Phi_j(x)=\Phi(x;\lambda_j)$, where (convergent improper integral)
\begin{equation}\label{Phij_int}
\Phi(x;\lambda)
= - \frac{i}{\epsilon} e^{ih(x;\lambda)/\epsilon } \int_{-\infty}^{x} e^{-ih(z;\lambda)/\epsilon }  \sum_{p=1}^{P} \frac{\phi_p(\lambda)}{z-z_p} \, dz.
\end{equation}
Here, the coefficients $\phi_p(\lambda)$ are defined for all $\lambda\in\mathbb{C}\setminus\mathbb{R}^+$ by 
\begin{equation}
\sum_{p=1}^PA_{mp}(\lambda)\phi_p(\lambda)=0,\quad m=1,\dots,P-1
\end{equation}
together with
\begin{equation}\label{Phi_norm}
 \sum_{p=1}^{P} \phi_p(\lambda) = \lambda.
\end{equation}
When $\lambda=\lambda_j<0$ is an eigenvalue, $\mathbf{A}(\lambda)$ has rank $P-1$, and thus 
$\boldsymbol{\phi}(\lambda)=(\phi_1(\lambda),\dots,\phi_P(\lambda))^\mathsf{T}$ is the unique nullvector of $\mathbf{A}(\lambda_j)$ normalized by the condition \eqref{Phi_norm}.

The scattering data associated with the rational potential $u_0$ are as follows \cite{MillerWetzel2015}.
\paragraph{Continuous spectrum for $\lambda>0$} 
\begin{itemize}
\item Reflection coefficient $\beta(\lambda)$ for $\lambda > 0$: defined by 
\begin{equation}\label{base_reflect_coeff}
\beta(\lambda) := \frac{i}{\epsilon} e^{iM/\epsilon} \int_{-\infty}^{\infty}  e^{-i h(z;\lambda)/\epsilon }  \left( u_0(z) - \sum_{p=1}^{P} \frac{v_p(\lambda)}{z-z_p}\right) \; dz
\end{equation}
with $M:=2\pi i(c_1+\cdots + c_P)\in\mathbb{R}$. The coefficients $\{v_p(\lambda)\}_{p=1}^{P}$, assembled in a vector $\mathbf{v}(\lambda) := (v_1(\lambda),\dots,v_P(\lambda))^\mathsf{T}$, satisfy the linear algebra problem 
\begin{equation}\label{lin_sys_Jost}
\mathbf{A}^{>}(\lambda) \mathbf{v}(\lambda) = \mathbf{b}^{>}(\lambda), \quad\lambda>0,
\end{equation}
where  $\mathbf{A}^{>}(\lambda)\in\mathbb{C}^{P\times P}$ and $\mathbf{b}^{>}(\lambda)\in\mathbb{C}^{P}$ denote the boundary values taken by $\mathbf{A}(\lambda)$ and $\mathbf{b}(\lambda)$, respectively, on $\mathbb{R}^+$ from the upper-half $\lambda$-plane.
\end{itemize}
\paragraph{Discrete spectrum for $\lambda<0$}
\begin{itemize}
\item Negative eigenvalues $\{\lambda_j\}_{j=1}^{N}$: the zeros of the Evans function 
$D(\lambda)$ defined by \eqref{Evans_function}.
\item Phase constants $\{\gamma_j\}_{j=1}^{N}$: defined by
\begin{equation}\label{Phase_const1}
\gamma_j  := \frac{\epsilon}{2 \pi \lambda_j} \int_{-\infty}^{\infty} \Phi_j(x)^{*} \left( x  \Phi_j(x) - 1 \right) \; dx
\end{equation}
 in terms of the corresponding eigenvalue $\lambda_j$ and normalized eigenfunction $\Phi_j(x)=\Phi(x;\lambda_j)$ defined by \eqref{Phij_int}.
 \end{itemize}

\begin{rmk}
The number $N$ of eigenvalues $\lambda_j$ and phase constants $\gamma_j$ is dependent on the dispersion parameter $\epsilon$. 
In fact, $N = N(\epsilon) = O(\epsilon^{-1})$ as $\epsilon \to 0$ for $P$ fixed; see Section~\ref{subsubsec:evansfunc}.
\end{rmk}

\begin{rmk}
The formula \eqref{Phase_const1} is a relation known to hold for general real potentials $u_0$  \cite{DoktorovLeble2007}.  In \cite{MillerWetzel2015} a different formula for $\gamma_j$ was shown to be valid for rational potentials $u_0$.  While that formula is especially useful for exact calculations (see Section~\ref{sec:example}), the general formula \eqref{Phase_const1} turns out to be better suited to small-dispersion asymptotics; see Section~\ref{subsubsec:phase_const_asymp}.
\end{rmk}

\section{Rational Klaus-Shaw Potentials}\label{subsec:Klaus-Shaw}

The analysis of the BO scattering data in the small-dispersion limit is the most straightforward in the case that the rational potential $u_0$ is of \emph{Klaus-Shaw} type.
\begin{defn}\label{def:rKS}
A rational Klaus-Shaw (rKS) potential $u_0$ is a rational function of the form \eqref{u0_def_frac} for which there exists $x_\mathrm{c}\in\mathbb{R}$ such that $u_0$ is strictly monotone on the intervals $(-\infty,x_\mathrm{c})$ and $(x_\mathrm{c},+\infty)$ (making $x_\mathrm{c}$ the unique real critical point).
\end{defn}

Rational Klaus-Shaw potentials are of one sign, and hence there are distinct classes of positive and negative rKS potentials.  Colloquially one can describe a rKS potential as having a graph consisting of a single ``lobe'' or ``bump,'' 
a property that was found by Klaus and Shaw \cite{KlausShaw2002} to be useful in confining the spectrum of certain non-self-adjoint operators.  This property also provides an important simplification in various asymptotic calculations involving the WKB approximation for differential equations with Klaus-Shaw coefficients, for which there are at most two real turning points.  The utility of the monotonicity condition is similar in the context of the BO scattering problem in the small-dispersion limit.  Indeed, for a rKS potential $u_0$ there are at most two real roots $\lambda\in\mathbb{R}$ of the equation $u_0(x) = -\lambda$, and hence the corresponding inverse function $x(\lambda)$ has two real branches. 
This implies some useful properties of the exponent function $h$ defined in \eqref{def_h}, because the critical points of $h$ satisfy
\begin{equation}\label{crit_point_eqn}
h'(x;\lambda) = \lambda + u_0(x) =0 \quad \text{or} \quad u_0(x) = -\lambda;
\end{equation}
here the prime denotes a derivative in $x$.

\begin{defn}\label{defn:bulk}
The bulk $\mathcal{B}\subset\mathbb{R}$ is the set $\mathcal{B} = \{\lambda: -\sup \{ u_0\} < \lambda < - \inf\{u_0\}\}$.  
We say that $\lambda\in\mathcal{B}$ lies in the bulk, while $\lambda\in\mathbb{R}\setminus\overline{\mathcal{B}}$ lies outside the bulk, where $\overline{\mathcal{B}}$ denotes the closure of $\mathcal{B}$.  
\end{defn}

\begin{rmk}\label{rmk:pos_neg_KlausShaw}
For a strictly positive (resp., negative) rKS potential $u_0$, $\mathcal{B} = (-\sup\{u_0\},0)$ (resp., $\mathcal{B}=(0,-\inf\{u_0\})$) so the bulk is a negative (resp., positive) interval abutting the origin.  We refer to the origin as the ``hard edge'' of the bulk and to the nonzero endpoint of $\mathcal{B}$ as the ``soft edge'' of the bulk.
\end{rmk}

\begin{defn}[Critical points for $\lambda$ in the bulk]\label{crit_point_Bulk}
Let $u_0$ be a rKS potential with corresponding bulk $\mathcal{B}$.  The functions $x_{\pm}: \mathcal{B} \longrightarrow \mathbb{R}$ represent the two real branches of the inverse of $-u_0$ with $x_{-}(\lambda)< x_{+}(\lambda)$ for $\lambda \in \mathcal{B}$ and are the only real critical points of the function $h$; see Figure~\ref{fig:u0_branches}. The remaining $2P-2$ critical points of $h$ form complex-conjugate pairs for $\lambda\in\mathcal{B}$ and will be denoted by $x_p,x_p^*$ with $\I\{x_p\}>0$ for $p = 1, \ldots, P-1$ (not necessarily distinct for all $\lambda\in\mathcal{B}$).
\end{defn}

\begin{figure}[h!]
\begin{center}
\includegraphics[scale=1]{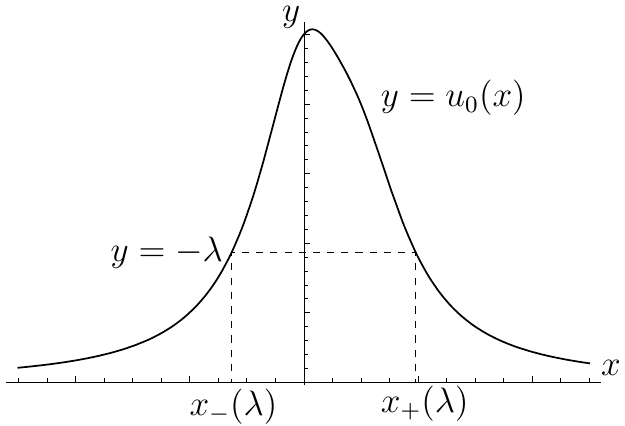}
\caption{A strictly positive rKS potential $u_0$ illustrating the two real branches $x_{\pm}(\lambda)$ of the inverse function of $-u_0$ defined for suitable $\lambda<0$.  For a strictly negative rKS potential the picture is similar, but $\lambda>0$ is necessary for the functions $x_\pm(\lambda)$ to be defined.
 } \label{fig:u0_branches}
\end{center}
\end{figure}

\begin{defn}[Critical points for $\lambda$ outside the bulk]\label{crit_point_notBulk}
Let $u_0$ be a rKS potential with corresponding bulk $\mathcal{B}$.  If $\lambda\in\mathbb{R}\setminus\overline{\mathcal{B}}$, then $h$ has $2P$ critical points in complex-conjugate pairs denoted by $x_p,x_p^*$ with $\I\{x_p\}>0$ for $p = 0, \ldots, P-1$ (not necessarily distinct for all $\lambda\in\mathbb{R}\setminus\overline{\mathcal{B}}$).
\end{defn}

In a neighborhood of $\lambda\in\mathcal{B}$ or $\lambda\in\mathbb{R}\setminus{\mathcal{B}}$ for which all $2P$ critical points of $h$ are simple, each critical point is a locally analytic function $x=x(\lambda)$ of $\lambda$ that satisfies \eqref{crit_point_eqn}.  As $\lambda$ exits the bulk through the soft edge, the two real critical points $x_\pm(\lambda)$ coalesce and bifurcate into the complex plane, where they become renamed as the (nearly real) conjugate pair $x_0(\lambda)$ and $x_0(\lambda)^*$.  As $\lambda$ exits the bulk through the hard edge, the two real critical points $x_\pm(\lambda)$ tend to $\pm\infty$ and then become finite again as a (large) conjugate pair $x_0(\lambda)$ and $x_0(\lambda)^*$.

Since according to \eqref{u0_def_frac} and \eqref{crit_point_eqn} $h'(x;\lambda)$ is a rational function of $x$ with simple poles at $\{z_p,z_p^*\}_{p=1}^P$ such that $h'(x;\lambda) \to \lambda$ as $x \to \infty$, using Definition~\ref{crit_point_Bulk} we may explicitly write
\begin{equation}\label{hprime_lambda}
h'(x;\lambda) = \lambda \Psi^{-}(x;\lambda) \Psi^{+}(x;\lambda) \quad \text{for} \quad \lambda \in \mathcal{B},
\end{equation}
where, assuming the complex critical points $x_p=x_p(\lambda)$ are distinct for $\lambda\in\mathcal{B}$,
\begin{equation}\label{Psi_defn}
\begin{split}
\Psi^{-}(x;\lambda) &:=  \left(x-x_{-}(\lambda) \right) \frac{ \prod_{p=1}^{P-1} \left( x-x_p(\lambda) \right) }{ \prod_{p=1}^{P} \left( x-z_p \right)}\\
\Psi^{+}(x;\lambda) &:=  \left(x-x_{+}(\lambda) \right) \frac{ \prod_{p=1}^{P-1} \left( x-x_p(\lambda)^* \right) }{ \prod_{p=1}^{P} \left( x-z_p^* \right)}.
\end{split}
\end{equation}

\begin{rmk}\label{rmk:analyticity_Psi}
Note that $\Psi^{+}$ ($\Psi^-$) is analytic in the upper-half (lower-half) $x$-plane.  For $\lambda$ outside the bulk, we may also write \eqref{hprime_lambda} provided we extend the definitions of $\Psi^{\pm}$ by replacing $x_-$ with  $x_0$ and $x_{+}$ with $x_0^{*}$ in \eqref{Psi_defn}. These extended definitions satisfy $\Psi^{-}(x;\lambda)^* = \Psi^{+}(x^*;\lambda)$ for all $\lambda\in\mathbb{R}\setminus\overline{\mathcal{B}}$.
\end{rmk}
Differentiating \eqref{hprime_lambda} and using the definitions \eqref{Psi_defn}, we may express $h''(x_\pm(\lambda);\lambda)$ for $\lambda\in\mathcal{B}$ as
\begin{equation}\label{u_prime_poly}
h''(x_{\pm}(\lambda);\lambda) = u_0'(x_\pm(\lambda))=\pm\lambda \frac{\left| \Psi^{\mp}\left( x_{\pm}(\lambda);\lambda \right) \right|^2}{x_{+}(\lambda) - x_{-}(\lambda)} \quad \text{for} \quad \lambda \in \mathcal{B}.
\end{equation}
From equation \eqref{u_prime_poly} we may immediately observe that 
\begin{equation}
\sign \{h''(x_{\pm}(\lambda);\lambda) \}= \pm \sign \{\lambda\} \quad \text{for} \quad \lambda \in \mathcal{B}.
\end{equation}
 This fact can also be deduced from equation~\eqref{crit_point_eqn} and Figure~\ref{fig:u0_branches}. Namely, for a positive rKS potential $u_0$, we have $u_0'\left(x_{-}(\lambda) \right) >0$ and $u_0'\left(x_{+}(\lambda) \right)<0$ for a given $\lambda \in \mathcal{B}$, so $h''(x_{\pm}(\lambda);\lambda) = \mp |u_0'(x_{\pm}(\lambda))|$. 
Similarly, for strictly negative $u_0$, we have that $h''\left(x_{\pm}(\lambda);\lambda\right) = \pm |u_0'\left( x_{\pm}(\lambda) \right)|$  for $\lambda \in \mathcal{B}$.

We close this section by showing how to evaluate the function $h(x;\lambda)$ at its real critical points for $\lambda\in\mathcal{B}$.
\begin{lem}\label{fp_fm}
Let $\lambda \in \mathcal{B}$, and set $h_\pm(\lambda):=h(x_\pm(\lambda);\lambda)$.  Then
\begin{equation}
h_{+}(\lambda) =  \theta_{+}(\lambda) + M
\quad \text{and} \quad 
h_{-}(\lambda) =  \theta_{-}(\lambda),
\end{equation}
where
\begin{equation}
M:= \int_{-\infty}^{\infty} u_0 (y )\, dy\quad\text{and}\quad \theta_\pm(\lambda):=\int_0^\lambda x_\pm(\eta)\,d\eta
\label{h_theta_def}
\end{equation}
are well-defined because $u_0\in L^1(\mathbb{R})$.
\end{lem}

\begin{proof}
Because $f(x)$ is an antiderivative of $u_0(x)$ that vanishes as $x\to -\infty$,
\begin{equation}\label{fm_eqn}
f \left( x_{-}(\lambda) \right)  = \int_{-\infty}^{x_{-}(\lambda)} u_0 (y )\, dy 
= -\lambda x_{-}(\lambda) - \int_{-\infty}^{x_{-}(\lambda)} y u_0' (y ) \, dy,
\end{equation}
by integration by parts.
We rewrite the final integral on the right-hand side using the substitution $u_0(y) = -\eta$, where $y = x_{-}(\eta)$ for $y \in (-\infty,x_{-}(\lambda)]$. Hence equation~\eqref{fm_eqn} becomes
\begin{equation}
f \left( x_{-}(\lambda) \right) 
= -\lambda x_{-}(\lambda) + \int_{0}^{\lambda} x_{-}(\eta) \, d\eta,
\end{equation}
i.e., $h_-(\lambda)=\theta_-(\lambda)$.
Starting instead from the formula
\begin{equation}
f \left( x_{+}(\lambda) \right) 
=  \int_{-\infty}^{\infty} u_0 (y )\, dy  - \int_{x_{+}(\lambda)}^{\infty} u_0 (y )\, dy
\end{equation}
but making the substitution $y=x_+(\eta)$ after integrating by parts in the second term then gives the desired formula for $h_+(\lambda)$.
\end{proof}

\section{Continuous Spectrum $\lambda>0$}\label{section-reflection-asymptotics}

An asymptotic expansion for the reflection coefficient $\beta(\lambda)$, $\lambda > 0$, can be directly obtained from formula \eqref{base_reflect_coeff} and the linear system \eqref{lin_sys_Jost}.
We first prove that the solution $\mathbf{v}(\lambda)$ of the linear system \eqref{lin_sys_Jost} has a limit $\mathbf{v}^0(\lambda)$ as $\epsilon\to 0$, at which point it only remains to apply the method of stationary phase to the integral in \eqref{base_reflect_coeff} to determine the leading-order behavior  of $\beta(\lambda)$.  

Every term in the $m^\text{th}$ equation of the linear system \eqref{lin_sys_Jost} is an integral over a common contour $\tilde{C}_m$ that is obtained from $C_m$ (appearing in the definition of the system
\eqref{lin_sys_Abar}) by rotation of the vertical tails through $\pi$ radians.  The multi-valued factor $e^{-ih(z;\lambda)/\epsilon}$ appearing in the integrand is analytically continued through its vertical branch cuts as the contours rotate so that when $\lambda>0$, $\mathbf{A}^>(\lambda)$ and $\mathbf{b}^>(\lambda)$ are the boundary values taken on $\mathbb{R}^+$ from the upper half-plane of the corresponding quantities $\mathbf{A}(\lambda)$ and $\mathbf{b}(\lambda)$ analytic for $\lambda\in\mathbb{C}\setminus\mathbb{R}^+$.  We may then define a related linear system in which each row is first taken as a linear combination over $\mathbb{Z}$ of the rows of \eqref{lin_sys_Jost} and then
simplifications accounting for factors arising from continuation of $e^{-ih(z;\lambda)/\epsilon}$ about its singularities $\{z_p\}_{p=1}^P$ are made.  This new system has the same form as \eqref{lin_sys_Jost} but each integral over $\tilde{C}_m$ is replaced by an integral of some analytic branch of the same integrand over a new contour which we denote by $W_m$, and it can be written in the form $\mathbf{N}\mathbf{A}^>(\lambda)\mathbf{v}(\lambda)=\mathbf{N}\mathbf{b}^>(\lambda)$ for some square matrix $\mathbf{N}$.  Our asymptotic analysis of the solution $\mathbf{v}(\lambda)$ of the system \eqref{lin_sys_Jost} will be effective provided the new system is \emph{suitable} in the sense that
(i) each contour $W_m$ passes through exactly one critical point of $h(z;\lambda)$ such that $\R\{-ih(z;\lambda)\}$ is maximized along $W_m$ at that point, (ii) the association of critical points to contours $W_m$ is one-to-one, and (iii) the matrix $\mathbf{N}$ is invertible (this property guarantees the equivalence of the old and new systems).

We believe that in general a suitable linear system can be found.  We present the following proposition as proof of principle in a general case.
\begin{prop}
\label{prop-Reflection-contours}
Let $u_0$ be a negative rKS potential for which $c_p$ is a positive imaginary number for all $p=1,\dots,P$, and let $\lambda>0$ be such that the complex critical points of $h(\cdot;\lambda)$ are all simple and correspond to distinct nonzero values of $\R\{-ih(z;\lambda)\}$.  Then there exists a family of contours $\{W_m\}_{m=1}^P$ such that the modified system
$\mathbf{N}\mathbf{A}^>(\lambda)\mathbf{v}(\lambda)=\mathbf{N}\mathbf{b}^>(\lambda)$ having
integrals over contours $\{W_m\}_{m=1}^P$ in place of $\{\tilde{C}_m\}_{m=1}^P$ is suitable.
\end{prop}
The condition on $\{c_p\}_{p=1}^P$ in Proposition~\ref{prop-Reflection-contours} implies that $u_0$ is a negative linear combination of Lorentzian profiles $2v_p/((x-u_p)^2+v_p^2)$, $z_p=u_p+iv_p$.  The proof of Proposition~\ref{prop-Reflection-contours} can be found in Appendix~\ref{Appendix-Reflection}.  We now show that for suitable systems, the solution $\mathbf{v}(\lambda)$ converges to a well-defined limit as $\epsilon\to 0$.
\begin{prop}
Let $\lambda>0$ be given.  Suppose that the system \eqref{lin_sys_Jost} is equivalent via a matrix $\mathbf{N}$ to a suitable system $\mathbf{N}\mathbf{A}^>(\lambda)\mathbf{v}(\lambda)=\mathbf{N}\mathbf{b}^>(\lambda)$ having integration contours $\{W_m\}_{m=1}^P$ in its $P$ rows.  Then
\begin{equation}
\lim_{\epsilon\to 0}\mathbf{v}(\lambda)=\mathbf{v}^0(\lambda),\quad where
\end{equation}
\begin{equation}\label{W0_sol}
v_p^0(\lambda) = \lambda \lim_{z \to z_p} (z - z_p) \Psi^{-}(z;\lambda) = \lambda \underset{z=z_p}{\Res} \Psi^{-}(z;\lambda),\quad p=1,\dots,P.
\end{equation}
\label{prop-limiting-coefficients-Jost}
\end{prop}
\begin{proof}
We give the proof in the case that all critical points associated with contours $\{W_m\}_{m=1}^P$ are simple, but the same result is true more generally.
Let $x_m$ denote the unique critical point of $-ih(z;\lambda)$ on $W_m$, which is traversed at the local steepest descent angle $\theta_m$, and define a diagonal matrix $\mathbf{D}$ whose elements are $D_{mm}:=|h''(x_m;\lambda)|^{1/2}e^{ih(x_m;\lambda)/\epsilon}e^{-i\theta_m}/\sqrt{2\pi\epsilon}$.  Set $\hat{\mathbf{A}}^>(\lambda):=\mathbf{D}\mathbf{N}\mathbf{A}^>(\lambda)$ and $\hat{\mathbf{b}}^>:=\mathbf{D}\mathbf{N}\mathbf{b}^>(\lambda)$.  Therefore
\begin{equation}
\begin{split}
\hat{A}^>_{mp}&=\sqrt{\frac{|h''(x_m;\lambda)|}{2\pi\epsilon}}e^{ih(x_m;\lambda)/\epsilon}e^{-i\theta_m}\int_{W_m}\frac{e^{-ih(z;\lambda)/\epsilon}}{z-z_p}\,dz\\
\hat{b}^>_{m}&=-\lambda\sqrt{\frac{|h''(x_m;\lambda)|}{2\pi\epsilon}}e^{ih(x_m;\lambda)/\epsilon}e^{-i\theta_m}\int_{W_m}e^{-ih(z;\lambda)/\epsilon}\,dz,
\end{split}
\end{equation}
so applying the method of steepest descent shows that $\hat{\mathbf{A}}^>(\lambda)\to\hat{\mathbf{A}}^{>0}(\lambda)$ and $\hat{\mathbf{b}}^>(\lambda)\to\hat{\mathbf{b}}^{>0}(\lambda)$ as $\epsilon\to 0$, where
\begin{equation}
\hat{A}^{>0}_{mp}(\lambda):=\frac{1}{x_m-z_p}\quad\text{and}\quad
\hat{b}^{>0}_m(\lambda):=-\lambda.
\end{equation}
The limiting matrix $\mathbf{A}^{>0}(\lambda)$ is of Cauchy type, and it is invertible precisely because the critical points $\{x_m\}_{m=1}^P$ associated with the contours $\{W_m\}_{m=1}^P$ are distinct.  
To show that $\mathbf{v}^0(\lambda)$ given by \eqref{W0_sol} is the (unique) solution of $\hat{\mathbf{A}}^{>0}(\lambda)\mathbf{v}^0(\lambda)=\hat{\mathbf{b}}^{>0}(\lambda)$, we may use the Residue Theorem; multiplying \eqref{W0_sol} on the left by $\hat{\mathbf{A}}^{>0}(\lambda)$ and
dividing by $-1/\lambda$ gives
\begin{equation}
- \frac{1}{\lambda} \sum_{p=1}^{P} \frac{v_p^0(\lambda)}{x_m - z_p} = 
\sum_{p=1}^{P} \frac{\underset{z=z_p}{\Res} \Psi^{-}(z;\lambda)}{z_p - x_m}
= \frac{1}{2 \pi i} \oint_C  \frac{\Psi^{-}(z;\lambda)}{z - x_m} \, dz = 1,
\end{equation}
where $C$ is a counter-clockwise contour encircling the points $\{z_p\}_{p=1}^{P}$ (the poles of the integrand) and the last equality follows by taking a residue at $z=\infty$ where the integrand behaves like $1/z$.  See \cite{Schechter59} for an alternative approach.
\end{proof}

\begin{rmk}
When the critical points are simple, $\hat{\mathbf{A}}^>(\lambda)-\hat{\mathbf{A}}^{>0}(\lambda)$ and
$\hat{\mathbf{b}}^>(\lambda)-\hat{\mathbf{b}}^{>0}(\lambda)$ both have complete asymptotic expansions in ascending integer powers of $\epsilon$.  It follows that the same is true of $\mathbf{v}(\lambda)-\mathbf{v}^0(\lambda)$.    If one or more of the critical points is not simple, fractional powers may appear in the expansions.
\label{rmk:power-series}
\end{rmk}

\begin{thm}\label{Reflect_asymptot}
Let $u_0$ be a rKS potential, and suppose that $\lambda>0$ is such that there exists a suitable modification of the linear system \eqref{lin_sys_Jost}.  If $\lambda$ lies outside the bulk (true for all $\lambda>0$ if $u_0$ is a positive rKS potential) then $\beta(\lambda)$ is exponentially small as $\epsilon\to 0$.  If $\lambda$ lies in the bulk, then
\begin{equation}\label{beta_asymp_form}
\beta(\lambda) 
= - \sqrt{\frac{2 \pi \lambda \left(x_{+}(\lambda) - x_{-}(\lambda) \right)}{\epsilon }} e^{- \frac{i}{\epsilon} \theta_{+}(\lambda) - i \psi_{+}(\lambda)} + o(\epsilon^{-1/2})
\end{equation}
as $\epsilon \to 0$, where $\theta_{+}$ is defined in \eqref{h_theta_def} and $\psi_{+}$ is defined by
\begin{equation}\label{psi_def_pos}
e^{\mp i \psi_{\pm}(\lambda) - i \pi /4} := \frac{\Psi^{\mp}\left( x_{\pm}(\lambda);\lambda \right)}{\left| \Psi^{\mp}\left( x_{\pm}(\lambda);\lambda \right)\right|}
\end{equation}
with $\Psi^{\pm}$ given in \eqref{Psi_defn}.  The error term in \eqref{beta_asymp_form} can be written as $O(\epsilon^{1/2})$ if the critical points of $h(z;\lambda)$ are all simple.
\end{thm}

\begin{proof}
By Proposition~\ref{prop-limiting-coefficients-Jost}, the coefficients $v_p(\lambda)$ appearing in \eqref{base_reflect_coeff} tend to limits $v_p^0(\lambda)$ as $\epsilon\to 0$.  Thus it remains to analyze the absolutely convergent integral over $\mathbb{R}$ of $e^{-ih(z;\lambda)/\epsilon}u_0(z)$ and convergent improper integrals over $\mathbb{R}$ of $e^{-ih(z;\lambda)/\epsilon}(z-z_p)^{-1}$ for $p=1,\dots,P$.

For $\lambda\in\mathbb{R}^+\setminus\overline{\mathcal{B}}$, there are no real critical points of $h(z;\lambda)$, and hence the contour of integration in \eqref{base_reflect_coeff} may be deformed away from the real axis in the direction of decrease of $\R\{-ih(z;\lambda)\}$ until the nearest critical point or singular point is reached, which proves the exponential decay of $\beta(\lambda)$ as $\epsilon\to 0$ (the distance of the nearest critical/singular point to $\mathbb{R}$ determines the exponential rate of decay).

For $\lambda\in\mathcal{B}$ there are exactly two simple stationary phase points at $z=x_\pm(\lambda)$, so by the method of stationary phase, if $g(z)=u_0(z)$ or $g(z)=(z-z_p)^{-1}$,
\begin{equation}
\frac{1}{\sqrt{2\pi\epsilon}}\int_{-\infty}^\infty e^{-ih(z;\lambda)/\epsilon}g(z)\,dz=\frac{e^{i\pi/4}e^{-ih_-(\lambda)/\epsilon}g(x_-(\lambda))}{\sqrt{|h''(x_-(\lambda);\lambda))|}} + 
\frac{e^{-i\pi/4}e^{-ih_+(\lambda)/\epsilon}g(x_+(\lambda))}{\sqrt{|h''(x_+(\lambda);\lambda)|}} + O(\epsilon)
\end{equation}
where $h_\pm(\lambda):=h(x_\pm(\lambda);\lambda)$, and we used the fact that as $u_0$ is necessarily a negative rKS potential (to have positive $\lambda\in\mathcal{B}$) we have $\sign(u_0'(x_\pm(\lambda)))=\pm 1$.
Therefore,
\begin{equation}\label{eq:reflection-coeff-expansion-1}
-ie^{-iM/\epsilon}\sqrt{\frac{\epsilon}{2\pi}}\beta(\lambda)=\frac{e^{i\pi/4}e^{-ih_-(\lambda)/\epsilon}G^-_\epsilon(\lambda)}{\sqrt{|h''(x_-(\lambda);\lambda)|}} +\frac{e^{-i\pi/4}e^{-ih_+(\lambda)/\epsilon}G^+_\epsilon(\lambda)}{\sqrt{|h''(x_+(\lambda);\lambda)|}} + O(\epsilon),
\end{equation}
where, using $u_0(x_\pm(\lambda))=-\lambda$,
\begin{equation}
G^\pm_\epsilon(\lambda):=-\lambda-\sum_{p=1}^P\frac{v_p(\lambda)}{x_\pm(\lambda)-z_p}.
\end{equation}
The $\epsilon$-dependence in $G^\pm_\epsilon(\lambda)$ enters through the coefficients $\{v_p(\lambda)\}_{p=1}^P$; we may replace $G^\pm_\epsilon(\lambda)$ with $G^\pm_0(\lambda)$ (in which $v_p(\lambda)$ is replaced with $v_p^0(\lambda)$ given by \eqref{W0_sol}) in \eqref{eq:reflection-coeff-expansion-1} provided the error term is replaced with $o(1)$ (unless the critical points of $h$ are all simple; see Remark~\ref{rmk:power-series}).  Now using \eqref{W0_sol}, 
we apply the Residue Theorem to obtain
\begin{equation}
G_0^\pm(\lambda)+\lambda=-\lambda\sum_{p=1}^P\underset{z=z_p}{\Res}\frac{\Psi^-(z;\lambda)}{x_\pm(\lambda)-z}=-\frac{\lambda}{2\pi i}\oint_C\frac{\Psi^-(z;\lambda)}{x_\pm(\lambda)-z}\,dz
\end{equation}
where $C$ encircles the poles $z_1,\dots,z_P$ positively but excludes the points $x_\pm(\lambda)$.
Now $\Psi^-(z;\lambda)$ includes a factor of $z-x_-(\lambda)$, so to compute $G_0^-(\lambda)+\lambda$ the singularity of the integrand at $z=x_-(\lambda)$ outside of $C$ is removable, so as the integrand behaves like $-1/z$ as $z\to\infty$ we obtain $G_0^-(\lambda)+\lambda=\lambda$ or simply $G_0^-(\lambda)=0$.  To compute $G_0^+(\lambda)+\lambda$ we have the same behavior of the integrand for large $z$ but now there is a simple pole outside $C$ at $z=x_+(\lambda)$, so taking it into account gives $G_0^+(\lambda)+\lambda=\lambda-\lambda\Psi^-(x_+(\lambda);\lambda)$ or simply 
$G_0^+(\lambda)=-\lambda\Psi^-(x_+(\lambda);\lambda)$.  Therefore, using \eqref{psi_def_pos}, \eqref{eq:reflection-coeff-expansion-1} may be written as
\begin{equation}
\beta(\lambda)=-\sqrt{\frac{2\pi}{\epsilon}}\frac{\lambda |\Psi^-(x_+(\lambda);\lambda)|}{\sqrt{|h''(x_+(\lambda);\lambda)|}}e^{iM/\epsilon}e^{-ih_+(\lambda)/\epsilon}e^{-i\psi_+(\lambda)} + o(\epsilon^{-1/2}).
\end{equation}
The phase is then simplified using Lemma~\ref{fp_fm}, and the amplitude is simplified by \eqref{u_prime_poly}.
\end{proof}

\begin{cor}[Matsuno's modulus formula \cite{Matsuno82}]\label{beta_mag_asymp}
Let $u_0$ be a negative rKS potential.  If $\lambda \in \mathcal{B}$ then
\begin{equation}
\lim_{\epsilon\to 0} \epsilon |\beta(\lambda) |^2
=2 \pi \lambda \left(x_{+}(\lambda) - x_{-}(\lambda) \right).
\end{equation}
\end{cor}

Matsuno obtained a result slightly more general than this (he omits the negative rKS condition and accordingly generalizes the limit formula) by an argument using trace formulae, identities that involve the modulus of $\beta$ but not its phase.  Theorem~\ref{Reflect_asymptot} both makes Matsuno's result rigorous for a class of potentials $u_0$ and also corrects and proves a conjecture in the thesis of Xu \cite{Xuthesis} regarding the phase of $\beta$.

\section{Discrete Spectrum $\lambda<0$}
\subsection{Evans Function and Eigenvalues}\label{subsubsec:evansfunc}

We study the asymptotic properties of the Evans function \eqref{Evans_function} for rKS potentials $u_0$ in the limit $\epsilon \to 0$ for $\lambda<0$;  according to 
\eqref{eq:eigenvalue-bound} it is sufficient to consider only \emph{positive} rKS potentials.
The essence of our strategy is to calculate the determinant in \eqref{Evans_function} for $\epsilon>0$ small by first applying the method of steepest descent to the individual entries \eqref{lin_sys_coeff} of the matrix $\mathbf{A}(\lambda)$, and then compute the leading term of the (finite) determinant $D(\lambda)$.  However, as in Section~\ref{section-reflection-asymptotics}, it is useful to apply this method not to $D(\lambda)=\det(\mathbf{A}(\lambda))$ but rather to $\tilde{D}(\lambda)=\det(\tilde{\mathbf{A}}(\lambda))$ where $\tilde{\mathbf{A}}(\lambda)=\mathbf{N}\mathbf{A}(\lambda)$ for some invertible $\mathbf{N}$. Otherwise, the computation of the determinant may produce cancellation to all orders of the steepest descent expansions of the matrix elements, which only means that the dominant contribution to $D(\lambda)$ arises from terms beyond all orders in those expansions i.e., terms arising from sub-dominant critical points.  Obviously $D(\lambda)$ and $\tilde{D}(\lambda)$ have exactly the same zeros, which is all that is important for obtaining the discrete spectrum.  Here we choose $\mathbf{N}$ so that the entries $\tilde{A}_{mp}(\lambda)$ of $\tilde{\mathbf{A}}(\lambda)$ have exactly the same form as \eqref{lin_sys_coeff} but with the contours $\{C_m\}_{m=1}^P$ replaced with others $\{W_m\}_{m=1}^P$ such that the resulting matrix $\tilde{\mathbf{A}}(\lambda)$ is suitable for asymptotic analysis in the sense described in Section~\ref{section-reflection-asymptotics} with one exception:  the contour of all integrals appearing in at most one row of $\tilde{\mathbf{A}}(\lambda)$ will be allowed to pass over either one critical point \emph{or two real critical points on the zero level of $\R\{-ih(z;\lambda)\}$}.
Such a choice of $\mathbf{N}$ has the precise effect of making the sub-dominant yet essential contributions to $D(\lambda)$ the dominant contributions to $\tilde{D}(\lambda)$. 

Again, as proof of principle that in a general case a suitable matrix $\tilde{\mathbf{A}}(\lambda)=\mathbf{N}\mathbf{A}(\lambda)$ can be found, we offer the following.
\begin{prop}\label{Wm_prop}
Let $u_0$ be a positive rKS potential for which $c_p$ is a negative imaginary number for all $p=1,\dots,P$, and let $\lambda<0$ be such that the complex critical points of $h(\cdot;\lambda)$ are all simple and correspond to distinct nonzero values of $\R\{-ih(z;\lambda)\}$.  Then there exists a family of contours $\{W_m\}_{m=1}^P$ such that the modified matrix $\tilde{\mathbf{A}}(\lambda):=\mathbf{N}\mathbf{A}(\lambda)$ having integrals over contours $\{W_m\}_{m=1}^P$ in place of $\{C_m\}_{m=1}^P$ is suitable and at most one of the contours, say $W_Q$, may pass over either one complex critical point or at most two real critical points $x_\pm(\lambda)$ for which $\R\{-ih(x_\pm(\lambda);\lambda)\}=0$.
\end{prop}
We prove this proposition 
in Appendix~\ref{Appendix-Evans}, and therein one can also find an elementary example for which $\mathbf{N}=\mathbb{I}$, i.e., the original matrix $\mathbf{A}(\lambda)$ is already suitable, and another for which $\mathbf{N}$ is nontrivial. 
Note that the condition on the coefficients $\{c_p\}_{p=1}^P$ in Proposition~\ref{Wm_prop} implies that $u_0$ is a positive linear combination of Lorentzian profiles $2v_p/((x-u_p)^2+v_p^2)$ where $z_p=u_p+iv_p$.  We now show how to calculate $\tilde{D}(\lambda)$ for a suitable matrix $\tilde{\mathbf{A}}(\lambda)$.

\begin{prop}
\label{prop-Dtilde-expand}
Let $u_0$ be a positive rKS potential, and let $\lambda<0$ be given.  Suppose that for some invertible matrix $\mathbf{N}$, the matrix $\tilde{\mathbf{A}}(\lambda):=\mathbf{N}\mathbf{A}(\lambda)$ is suitable.  If in addition $\lambda$ lies outside the bulk, then
\begin{equation}
\tilde{D}(\lambda)=U_\epsilon(\lambda)(1+O(\epsilon)),
\label{D-tilde-outside}
\end{equation}
where $U_\epsilon(\lambda)$ is a nonvanishing function defined by \eqref{eq:U-epsilon},
while for $\lambda$ in the bulk,
\begin{equation}
\tilde{D}(\lambda)=V_\epsilon(\lambda)\left(\sin(\pi T_\epsilon(\lambda))+O(\epsilon)\right)
\label{D-tilde-inside}
\end{equation}
where $V_\epsilon(\lambda)$ is a nonvanishing function defined by \eqref{eq:V-epsilon},
and 
\begin{equation}\label{Phase_e_defn}
T_\epsilon(\lambda):=
\frac{1}{2\pi\epsilon}\int_{-\max\{u_0\}}^\lambda(x_+(\eta)-x_-(\eta))\,d\eta + \frac{1}{2\pi}(\psi_+(\lambda)-\psi_-(\lambda))
\end{equation}
with $\psi_\pm(\lambda)$ being defined by \eqref{psi_def_pos}.  The error terms in \eqref{D-tilde-outside} and \eqref{D-tilde-inside} are uniform on compact subintervals of the indicated sets.
\end{prop}
Note that the formula \eqref{D-tilde-outside} is consistent with the bound \eqref{eq:eigenvalue-bound} on the discrete spectrum; there can be no eigenvalues outside the bulk $\mathcal{B}$, so it is impossible for $D(\lambda)$ (or equivalently $\tilde{D}(\lambda)$) to vanish there.  Note also that at the soft edge $\lambda_\mathrm{min}=-\max\{u_0\}$, we have $\psi_+(\lambda_\mathrm{min})-\psi_-(\lambda_\mathrm{min})=\pi/2$.
\begin{proof}
Let $Q$ be the index of the row of $\tilde{\mathbf{A}}(\lambda)$ whose integration contour may pass through one or two real critical points.  
Since $W_m$ contains exactly one dominant critical point $x_m=x_m(\lambda)\in\mathbb{C}^+$ for $m\neq Q$, the method of steepest descent gives
\begin{equation}
\label{Dm_asymp}
\tilde{A}_{mp}(\lambda)=\sqrt{\frac{2\pi\epsilon}{|h''(x_m;\lambda)|}}e^{i\theta_m}e^{-ih_m(\lambda)/\epsilon}\left(\frac{1}{x_m-z_p}+O(\epsilon)\right)
\end{equation}
as $\epsilon\to 0$, where $h_m(\lambda):=h(x_m(\lambda);\lambda)$, and $\theta_m$ is the steepest descent direction with which $W_m$ traverses the critical point $x_m$.  This same formula also holds for $m=Q$, provided $\lambda\in\mathbb{R}^-\setminus\overline{\mathcal{B}}$.  On the other hand, if $\lambda\in\mathcal{B}$, the contour $W_Q$ traverses two simple real critical points $x_\pm=x_\pm(\lambda)$ by the Klaus-Shaw condition, both of which contribute at the leading order to $\tilde{A}_{Qp}(\lambda)$:
\begin{multline}\label{DN_asymp}
\tilde{A}_{Qp}(\lambda) = 
\sqrt{\frac{2 \pi \epsilon}{|h''(x_{+}(\lambda);\lambda)|}} e^{i\pi/4}e^{-i h_{+}(\lambda)/\epsilon} \left( \frac{1}{x_{+}(\lambda)-z_p} + O(\epsilon)\right)\\
{}+  \sqrt{\frac{2 \pi \epsilon}{|h''(x_{-}(\lambda);\lambda)|}} e^{-i\pi/4}e^{-i h_{-}(\lambda)/\epsilon} \left( \frac{1}{x_{-}(\lambda) - z_p} + O(\epsilon)\right),
\end{multline}
where $h_\pm(\lambda)$ are defined in Lemma~\ref{fp_fm}.  Aside from the Cauchy factors $(x_m-z_p)^{-1}$ and $(x_\pm(\lambda)-z_p)^{-1}$ there is no dependence on the column index $p$ in these leading terms.  Therefore, by row-multilinearity of the determinant, if $\lambda\in\mathbb{R}^-\setminus\overline{\mathcal{B}}$, 
\begin{equation}
\tilde{D}(\lambda)=(2\pi\epsilon)^{P/2}\left(\prod_{m=1}^P\frac{e^{i\theta_m}e^{-ih_m(\lambda)/\epsilon}}{\sqrt{|h''(x_m;\lambda)|}}\right)\det\left(\mathbf{C}+O(\epsilon)\right),
\end{equation}
where $\mathbf{C}$ is the $P\times P$ Cauchy matrix $C_{mp}:=(x_m-z_p)^{-1}$.  
This matrix is invertible because the complex critical points $\{x_m\}_{m=1}^P$ are distinct, as are the poles $\{z_p\}_{p=1}^P$.  Hence if the nonvanishing function $U_\epsilon(\lambda)$ is defined 
by
\begin{equation}
\label{eq:U-epsilon}
U_\epsilon(\lambda):=(2\pi\epsilon)^{P/2}\left(\prod_{m=1}^P\frac{e^{i\theta_m}e^{-ih_m(\lambda)/\epsilon}}{\sqrt{|h''(x_m;\lambda)|}}\right)\det\left(\mathbf{C}\right),
\end{equation}
the proof of \eqref{D-tilde-outside} is complete.

On the other hand, if $\lambda\in\mathcal{B}$, then at the cost of a factor $(-1)^{P-Q}$ we may assume that $Q=P$ and we then split the last row before computing determinants to obtain
\begin{multline}
\tilde{D}(\lambda)=(2\pi\epsilon)^{P/2}\left(\prod_{m=1}^{P-1}\frac{e^{i\theta_m}e^{-ih_m(\lambda)/\epsilon}}{\sqrt{|h''(x_m;\lambda)|}}\right)\\
{}\cdot\left(\frac{e^{i\pi/4}e^{-ih_+(\lambda)/\epsilon}\det\left(\mathbf{C}^+\right)}{\sqrt{|h''(x_+(\lambda);\lambda)|}}
+
\frac{e^{-i\pi/4}e^{-ih_-(\lambda)/\epsilon}\det\left(\mathbf{C}^-\right)}{\sqrt{|h''(x_-(\lambda);\lambda)|}}+O(\epsilon)\right)
\end{multline}
where the matrix $\mathbf{C}^\pm$ is the Cauchy matrix $\mathbf{C}$ with the last row replaced by $((x_\pm(\lambda)-z_1)^{-1},\dots,(x_\pm(\lambda)-z_P)^{-1})$, and we used the fact that $|e^{-ih_\pm(\lambda)/\epsilon}|=1$.  Next we observe that \cite{Schechter59}
\begin{equation}
\begin{split}
\det\left(\mathbf{C}^+\right)&=P(\lambda)\frac{\Psi^-(x_+(\lambda);\lambda)}{x_+(\lambda)-x_-(\lambda)}\\
\det\left(\mathbf{C}^-\right)&=-P(\lambda)\frac{\Psi^+(x_-(\lambda);\lambda)^*}{x_+(\lambda)-x_-(\lambda)},
\end{split}
\end{equation}
$\Psi^\pm$ are defined by \eqref{Psi_defn} and where $P(\lambda)\neq 0$ is given by
\begin{equation}
P(\lambda):=\frac{\prod_{1\le i<j\le P-1}(x_j-x_i)\prod_{1\le i<j\le P}(z_i-z_j)}{\prod_{m=1}^{P-1}\prod_{p=1}^P(x_m-z_p)}.
\end{equation}
Using \eqref{u_prime_poly} and Lemma~\ref{fp_fm}, noting that because $u_0$ is positive Klaus-Shaw
\begin{equation}
\begin{split}
h_+(\lambda)-h_-(\lambda)&=\int_{-\infty}^\infty u_0(x)\,dx-\int_{\lambda}^0(x_+(\eta)-x_-(\eta))\,d\eta \\
&= 
\int_{-\max\{u_0\}}^\lambda(x_+(\eta)-x_-(\eta))\,d\eta,
\end{split}
\end{equation}
and defining (restoring the factor $(-1)^{P-Q}$)
\begin{equation}\label{eq:V-epsilon}
V_\epsilon(\lambda):=\frac{-2i(-1)^{P-Q}(2\pi\epsilon)^{P/2}P(\lambda)}{\sqrt{-\lambda (x_+(\lambda)-x_-(\lambda))}} \left(\prod_{m=1}^{P-1}\frac{e^{i\theta_m}e^{-ih_m(\lambda)/\epsilon}}{\sqrt{|h''(x_m;\lambda)|}}\right)
e^{-i(h_+(\lambda)+h_-(\lambda))/(2\epsilon)}e^{-i(\psi_+-\psi_-)/2}
\end{equation}
then establishes \eqref{D-tilde-inside} and completes the proof.
\end{proof}

\begin{rmk}
\label{remark-steepest-descent-lambda-perturb}
It is easy to see that if the fixed value $\lambda\in\mathcal{B}$ is replaced with an $\epsilon$-dependent value $\lambda_\epsilon$ with the property that for some fixed $\lambda\in\mathcal{B}$ and some constant $C>0$, $|\lambda-\lambda_\epsilon|\le C\epsilon$ holds for all $\epsilon>0$ sufficiently small,
then the asymptotic behavior of the matrix elements $\tilde{A}_{mp}(\lambda_\epsilon)$ can also be calculated by the method of steepest descent, with results similar to \eqref{Dm_asymp}--\eqref{DN_asymp}.  One need only replace the functions $h_m(\lambda)=h(x_m(\lambda);\lambda)$ and $h_\pm(\lambda)=h(x_\pm(\lambda);\lambda)$ in the exponents on the right-hand sides of \eqref{Dm_asymp}--\eqref{DN_asymp} by $h(x_m(\lambda);\lambda_\epsilon)$ and
$h(x_\pm(\lambda);\lambda_\epsilon)$ respectively.
\end{rmk}

Proposition~\ref{prop-Dtilde-expand} allows us to establish a number of important asymptotic properties of the discrete eigenvalues for positive rKS potentials.
\begin{thm}[Uniform approximation of eigenvalues]
\label{theorem-uniform}
Let $u_0$ be a positive rKS potential and suppose that for each $\lambda\in K\subset\mathcal{B}$, $K$ compact, a suitable modification $\tilde{\mathbf{A}}(\lambda)$ of $\mathbf{A}(\lambda)$ can be found.  
Then there is a constant $C_K$ such that for each eigenvalue $\lambda$ in $K$ there exists $\lambda_0<0$ satisfying $T_\epsilon(\lambda_0)\in\mathbb{Z}$, such that $|\lambda-\lambda_0|\le C_K\epsilon^2$ holds for all sufficiently small $\epsilon>0$.  Likewise, for each $\lambda_0\in K$ 
satisfying $T_\epsilon(\lambda_0)\in\mathbb{Z}$ there is an eigenvalue $\lambda$ such that the same estimate holds true.  
\end{thm}
\begin{proof}
The eigenvalues are characterized exactly by $\tilde{D}(\lambda)=0$, or using Proposition~\ref{prop-Dtilde-expand} in the case $\lambda\in\mathcal{B}$, $\sin(\pi T_\epsilon(\lambda))=O(\epsilon)$, the error term being uniform for $\lambda\in K\subset\mathcal{B}$.  Solving this equation for $T_\epsilon(\lambda)$ and multiplying by $\epsilon$ gives, for some $n\in\mathbb{Z}$,
\begin{equation}
\frac{1}{2\pi}\int_{-\min\{u_0\}}^\lambda (x_+(\eta)-x_-(\eta))\,d\eta +\frac{1}{2\pi}\epsilon (\psi_+(\lambda)-\psi_-(\lambda))=\epsilon n +O(\epsilon^2)
\end{equation}
uniformly on $K$.  Since the left-hand side is differentiable with respect to $\lambda$ with a derivative that is strictly positive on $K$, the result follows from the Implicit Function Theorem.
\end{proof}

\begin{cor}[Local approximation of eigenvalues]\label{cor_asymp_eig_loc}
Fix a closed interval $K\subset\mathcal{B}$ and a positive integer $J$.  Under the hypotheses of Theorem~\ref{theorem-uniform}, for each $\Lambda\in K$ the $2J+1$ eigenvalues $\lambda$ closest to $\Lambda$ are given by
\begin{equation}
\lambda=\Lambda+2 \pi \epsilon  \left(\frac{j+[T_\epsilon(\Lambda)]-T_\epsilon(\Lambda)}{x_{+}(\Lambda)-x_{-}(\Lambda)}\right) + O(\epsilon^2), \quad  j=-J,-(J-1),\dots,J-1,J,
\label{eq:eigenvalue-expansion}
\end{equation}
in the limit $\epsilon\to 0$, where the error term depends on $J$ and $K$, and where $[\cdot]$ denotes the nearest integer function.  Note that $\left|[T_\epsilon(\Lambda)]-T_\epsilon(\Lambda)\right| \leq 1/2$.
\end{cor}

\begin{proof}
We set $\lambda=\Lambda+\Delta$ and write the exact eigenvalue condition 
as $\epsilon T_\epsilon(\Lambda+\Delta)=\epsilon n+O(\epsilon^2)$ for some integer $n$.  Writing the integer $n$ in the form $n=[T_\epsilon(\Lambda)]+j$ for an integer $j\in [-J,J]$ and applying Taylor expansion,
\begin{equation}\label{eq:uniform_estimate}
\epsilon T_\epsilon'(\Lambda)\Delta + O(\Delta^2)=\epsilon\left(j+[T_\epsilon(\Lambda)]-T_\epsilon(\Lambda)\right) + O(\epsilon^2),
\end{equation}
because derivatives of $\epsilon T_\epsilon$ are bounded uniformly for $\lambda\in K$.  Now we seek solutions $\Delta=\lambda-\Lambda$ that are $O(\epsilon)$, so the error terms may be combined and we may solve for $\Delta$:
\begin{equation}
\lambda-\Lambda=\Delta=\frac{\epsilon(j+[T_\epsilon(\Lambda)]-T_\epsilon(\Lambda)) + O(\epsilon^2)}{\epsilon T'_\epsilon(\Lambda)}
\end{equation}
Since $\epsilon T'_\epsilon(\Lambda)=(x_+(\Lambda)-x_-(\Lambda))/(2\pi) + O(\epsilon)$ uniformly for $\Lambda\in K$ the proof is complete.
\end{proof}

This result shows that the eigenvalues are locally equally spaced with $O(\epsilon)$ spacing that depends on the point $\Lambda$ of local expansion.  Ignoring the details of the equal spacing and the offset of the grid given by the term $[T_\epsilon(\Lambda)]-T_\epsilon(\Lambda)$, we reproduce a result obtained by Matsuno \cite{Matsuno81} by formal asymptotic analysis of trace formulae (conservation laws).
\begin{cor}[Matsuno's density formula \cite{Matsuno81}]
\label{cor-Matsuno-density}
The asymptotic density of eigenvalues at a point $\lambda$ in the bulk is $\rho_\mathrm{M}(\lambda)/\epsilon$, where Matsuno's density is
\begin{equation}\label{Matsunos_density}
\rho_\mathrm{M}(\lambda):=\frac{1}{2 \pi}(x_+(\lambda)-x_-(\lambda)).
\end{equation}
In other words, the number $N[a,b]$ of eigenvalues in a subinterval $[a,b]$ of the bulk $\mathcal{B}$ satisfies
\begin{equation}
N[a,b]=\frac{1}{\epsilon}\int_a^b\rho_\mathrm{M}(\lambda)\,d\lambda + O(1),\quad\epsilon\to 0.
\end{equation}
\end{cor}
For all positive rKS potentials $u_0$, Matsuno's density $\rho_\mathrm{M}$ vanishes at the soft edge and blows up at the hard edge $\lambda=0$.  This should be contrasted with the better-known Weyl density of eigenvalues for semiclassical $1$-D Schr\"odinger operators, which is typically finite and nonzero at both hard and soft edges.
\begin{rmk}
In \cite{KodamaAblowitzSatsuma82} it was shown that if $u_0(x)=2/(1+x^2)$, then for $\epsilon=1/N$ for a positive integer $N$, the eigenvalues are precisely the roots of the equation $L_N(-2N\lambda)=0$ where $L_N$ is the Laguerre polynomial of degree $N$.  In this case, the distribution of eigenvalues is well-known in the theory of orthogonal polynomials and random matrix theory.  In the latter context Matsuno's density for the case $u_0(x)=2/(1+x^2)$ is known as the Mar\v{c}enko-Pastur law
\cite{MarchenkoP67}; from this context that we adapt the terminology of soft and hard edges of the (discrete) spectrum $\mathcal{B}$.
From the formula \eqref{base_reflect_coeff} one can see that for this potential, when $\epsilon=1/N$ the reflection coefficient $\beta$ vanishes identically.
\end{rmk}

\subsection{Phase Constants}\label{subsubsec:phase_const_asymp}

Fix a number $\Lambda\in\mathcal{B}$ and for each $\epsilon>0$ sufficiently small let $\lambda_\epsilon$ be the eigenvalue closest to $\Lambda$, choosing $\lambda_\epsilon>\Lambda$ if there are two equidistant.  According to Corollary~\ref{cor_asymp_eig_loc}, there is some constant $C>0$ depending only on $\Lambda$, such that $|\Lambda-\lambda_\epsilon|\le C\epsilon$ holds for all $\epsilon>0$ sufficiently small.  Denote by $\gamma_\epsilon$ and $\Phi_\epsilon(x)$ the phase constant and normalized eigenfunction, respectively, associated to the eigenvalue $\lambda_\epsilon$.  Then \eqref{Phase_const1} takes the form
\begin{equation}
\begin{split}
\gamma_\epsilon&=\frac{\epsilon}{2\pi\lambda_\epsilon}\int_{-\infty}^{\infty}\Phi_\epsilon(x)^*\left(x\Phi_\epsilon(x)-1\right)\,dx \\ {}&= \lim_{R\to +\infty}\left[\int_{-R}^R
xI_\epsilon(x)\,dx - \frac{\epsilon}{2\pi\lambda_\epsilon}\int_{-R}^R\Phi_\epsilon(x)^*\,dx\right]\\
{}&=\lim_{R\to+\infty}\int_{-R}^R xI_\epsilon(x)\,dx -\frac{i\epsilon}{2\lambda_\epsilon},\;\;
I_\epsilon(x):=\frac{\epsilon}{2\pi\lambda_\epsilon}|\Phi_\epsilon(x)|^2.
\end{split}
\label{gamma_eps}
\end{equation}
In the last step we deformed the contour $[-R,R]$ to a semicircle in the lower half-plane where $\Phi(z^*)^*$ is analytic and satisfies $\Phi(z^*)^*=z^{-1}+O(z^{-2})$ as $z\to\infty$.

\begin{prop}\label{proposition-I-epsilon}
Let $\Lambda\in\mathcal{B}$ be such that the complex critical points of $h(\cdot;\Lambda)$ are simple, and set $L=L(\Lambda):=1+\max\{|x_+(\Lambda)|,|x_-(\Lambda)|\}$.  The function $I_\epsilon(x)$ defined for $x\in\mathbb{R}$ in \eqref{gamma_eps} has the following properties.
\begin{enumerate}
\item There is a constant $K_0 = K_0(\Lambda)$ such that 
\begin{equation}
\left|I_\epsilon(x)-\frac{\epsilon}{2\pi\lambda_\epsilon x^2}\right| \leq \frac{K_0\epsilon}{|x|^3},\quad |x|\geq L
\label{I_eps_tail_bound}
\end{equation}
holds for all $\epsilon>0$ sufficiently small.  
\item There exists a constant $K_1 = K_1(\Lambda)$ such that 
\begin{equation}
|I_\epsilon(x)|\leq K_1,\quad |x|\le L
\label{I_eps_uniform_bound}
\end{equation}
holds for all $\epsilon>0$ sufficiently small.
\item For each real $x\neq x_\pm(\Lambda)$, 
\begin{equation}
\lim_{\epsilon\to 0} I_\epsilon(x)=I_0(x):=-\frac{\chi_{(x_-(\Lambda),x_+(\Lambda))}(x)}{x_+(\Lambda)-x_-(\Lambda)},
\label{I_eps_pointwise_limit}
\end{equation}
where $\chi_{(a,b)}(x)$ denotes the characteristic function of $(a,b)$.
\end{enumerate}
\end{prop}

\begin{rmk}
Comparing \eqref{I_eps_pointwise_limit} with the definition of $I_\epsilon(x)$ given in \eqref{gamma_eps} shows that for $x\in\mathbb{R}$ the modulus of the normalized eigenfunction $\Phi_{\epsilon}(x)$ behaves, in the limit $\epsilon \to 0$, like a large (of size $\epsilon^{-1/2}$) multiple of the characteristic function of the interval bounded by the turning points $x_\pm(\Lambda)$. This asymptotic behavior can be compared with the more familiar result, based on the WKB method, that eigenfunctions of a Sturm-Liouville problem, while also being asymptotically confined to a ``classically allowed'' interval $(x_-(\Lambda),x_+(\Lambda))$, nonetheless have a decidedly non-constant amplitude on this interval, diverging at the endpoints typically like $|x-x_\pm(\Lambda)|^{-1/4}$.  This draws a strong contrast between the nonlocal BO scattering problem and scattering problems for many other integrable equations like KdV involving spectral analysis of purely differential operators.
\end{rmk}

Proposition~\ref{proposition-I-epsilon} is proved in Appendix~\ref{Appendix-Phase}.
 Its main purpose is to prove the following theorem.
 \begin{thm}\label{theorem-gamma-limit}
For each $\Lambda\in\mathcal{B}$, 
\begin{equation}
\lim_{\epsilon \to 0} \gamma_{\epsilon} =  -\overline{x}(\Lambda),\quad\overline{x}(\Lambda):= \frac{1}{2} \left( x_{+}(\Lambda) + x_{-}(\Lambda) \right),
\label{equation-gamma-limit}
\end{equation}
where $\gamma_{\epsilon}$ is the phase constant for the eigenvalue $\lambda_\epsilon$ nearest  $\Lambda$.
\end{thm}

\begin{proof}
Since $\lambda_\epsilon\to\Lambda<0$ as $\epsilon\to 0$, from \eqref{gamma_eps} we see that
$\I\{\gamma_\epsilon\}\to 0$ as $\epsilon\to 0$, so it remains to analyze the real part.  Next, recalling $L=L(\Lambda)$ as defined in Proposition~\ref{proposition-I-epsilon} and using property~1 from the same yields, for each $R>L$,
\begin{equation}
\begin{split}
\left|\int_{\mathcal{I}_L^R} xI_\epsilon(x)\,dx\right|&\le
\left|\int_{\mathcal{I}_L^R}\frac{\epsilon\,dx}{2\pi\lambda_\epsilon x}\right| +\int_{\mathcal{I}_L^R}|x|\left|I_\epsilon(x)-\frac{\epsilon}{2\pi\lambda_\epsilon x^2}\right|\,dx\\
&=\int_{\mathcal{I}_L^R}|x|\left|I_\epsilon(x)-\frac{\epsilon}{2\pi\lambda_\epsilon x^2}\right|\,dx\\
&\le\int_{\mathcal{I}_L^R}\frac{K_0\epsilon\,dx}{x^2} \\
&\le \frac{2K_0\epsilon}{L}, \quad \mathcal{I}_L^R:=[-R,-L]\cup [L,R],
\end{split}
\end{equation}
an upper bound that is independent of $R>L$ and that tends to zero with $\epsilon$.  Hence
\begin{equation}
\lim_{\epsilon\to 0}\gamma_\epsilon = \lim_{\epsilon\to 0}\int_{-L}^L x I_\epsilon(x)\,dx,
\end{equation}
so applying the Lebesgue Dominated Convergence Theorem to calculate the latter limit, making use 
of properties~2 (integrable $\epsilon$-independent domination) and 3 (pointwise limit), gives
\begin{equation}
\lim_{\epsilon \to 0} \gamma_\epsilon = \int_{-L}^{L} x   I_{0}(x) \, dx =  
-\int_{x_-(\Lambda)}^{x_+(\Lambda)}\frac{x\,dx}{x_+(\Lambda)-x_-(\Lambda)}.
\end{equation}
Evaluating the integral then proves \eqref{equation-gamma-limit}.
\end{proof}

\begin{rmk}
Theorem~\ref{theorem-gamma-limit} shows that the limit of $\gamma_{\epsilon}$ is a purely real number. On the other hand, $\I \{\gamma_\epsilon\} = -i \epsilon/(2 \lambda_{\epsilon})$ cannot always be neglected outright.  For example, it plays a pivotal role in the small-dispersion analysis of certain determinantal $\tau$-functions arising in inverse-scattering theory \cite{MillerXu2011}.
\end{rmk}

A version of the asymptotic formula \eqref{equation-gamma-limit} characterizing the real part of $\gamma_\epsilon$ for small $\epsilon$ was hypothesized, based on asymptotic analysis of the BO inverse scattering problem, in \cite{MillerXu2011}.  Theorem~\ref{theorem-gamma-limit} provides a careful statement and is the first rigorous result on the asymptotic behavior of the phase constants, which (like the phase of the reflection coefficient) cannot be captured via trace formulae.

\section{Numerical Verification}\label{sec:example}

We illustrate the accuracy of our asymptotic formulae for the scattering data by comparing them to exact calculations in the case of the rKS potential:
\begin{equation}\label{example_IC}
u_0(x) = -\nu i \left(\frac{2}{x-i} + \frac{1}{x-(i+1)}\right) + c.c.
\end{equation}
for $\nu = \pm 1$.  The graph of $u_0$ is plotted in Figure~\ref{fig:u0_branches} for $\nu=1$, which confirms the Klaus-Shaw condition. 
This rKS potential is positive when $\nu = 1$ (negative when $\nu =-1$).
The bulk $\mathcal{B}$ consists of the interval $ (-\max\{u_0\},0)$ when $\nu =1$ with $\max\{u_0\} \approx 5.07308$ and $(0,-\min \{u_0\})$ when $\nu = -1$ with $\min\{u_0\} \approx - 5.07308$; see Section~\ref{subsec:Klaus-Shaw} and Definition~\ref{defn:bulk}.
We selected this rKS potential so that it is not even about any point, a property making the phase constants $\gamma_j$ nontrivial to calculate and interesting to compare with small-dispersion asymptotics. 

We compute the exact eigenvalues $\lambda_j$ for the potential \eqref{example_IC} with $\nu =1$ using the Evans function \eqref{Evans_function}.
To simplify computations only values of $\epsilon>0$ for which $ic_1/\epsilon\in\mathbb{Z}$ and $ic_2/\epsilon\in\mathbb{Z}$ are considered.
Since $ic_1$ and $ic_2$ are integers, this requires $\epsilon =1/m$ for some $m \in \mathbb{N}$. In this case, the integrals \eqref{lin_sys_Abar} defining the elements of the matrix $\mathbf{A}(\lambda)$ can be calculated explicitly by the Residue Theorem and the exact eigenvalues are thus obtained as the roots of a polynomial; see \cite{MillerWetzel2015}.

In Figure~\ref{fig:eigs_eps} we show the exact eigenvalues (black dots) and compare them (i) with  their uniform approximations obtained from Theorem~\ref{theorem-uniform} by solving $T_\epsilon(\lambda)=n$ for positive integers $n$ (overlaid squares) and (ii) with their local equally-spaced approximations described by \eqref{eq:eigenvalue-expansion} in Corollary~\ref{cor_asymp_eig_loc} (overlaid circles), under the scaling $(\lambda-\Lambda)/\epsilon$. 
As expected, the uniform approximation's squares track the black dots very well, while the circles do so best near the point $\lambda = \Lambda$. 
 In Figure~\ref{fig:eigs_hist} we present histograms of the exact eigenvalues to highlight their distribution on $\mathbb{R}^{-}$. The histograms clearly match better and better, in the limit $\epsilon\to 0$, the density $\rho_\mathrm{M}(\lambda)$ of Matsuno (Corollary~\ref{cor-Matsuno-density}), here normalized to have  integral equal to unity.

\begin{figure}[h!]
\begin{center}
\includegraphics[scale=1]{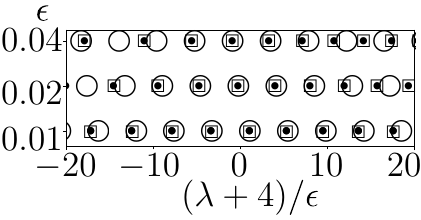}
\includegraphics[scale=1]{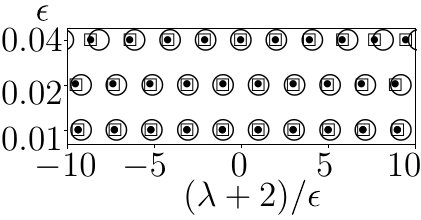}\\
\includegraphics[scale=1]{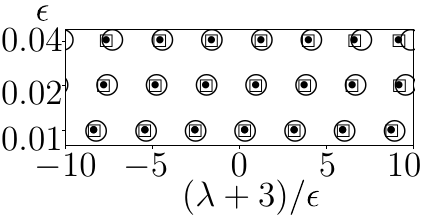}
\includegraphics[scale=1]{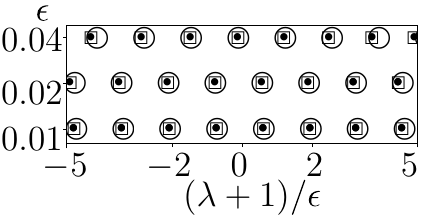}
\caption{
Exact eigenvalues for the rKS potential \eqref{example_IC} with $\nu = 1$ as a function of $\epsilon$ (Black dots) with the uniform approximations (squares) and the local approximations based at various values of $\Lambda\in\mathcal{B}$ (circles) overlaid. 
Note that the horizontal axis is the rescaled local coordinate $(\lambda-\Lambda)/\epsilon$ with $\Lambda=-4$ (top left), $\Lambda =-3$ (top right), $\Lambda=-2$ (bottom left), and $\Lambda =-1$ (bottom right).
} \label{fig:eigs_eps}
\end{center}
\end{figure}

\begin{figure}[h!]
\begin{center}
\includegraphics[scale=1]{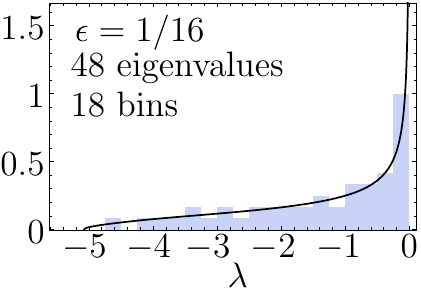}
\includegraphics[scale=1]{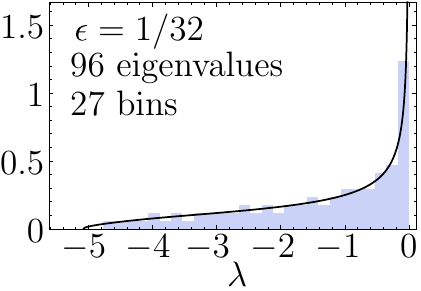}\\
\includegraphics[scale=1]{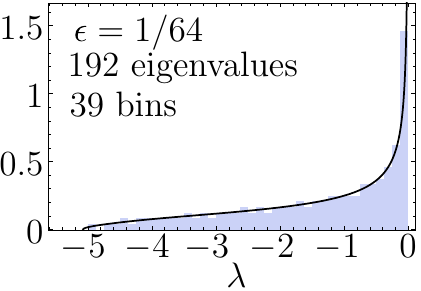}
\includegraphics[scale=1]{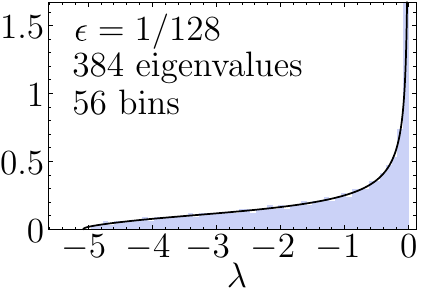}
\caption{Histograms of exact eigenvalues for the rKS potential \eqref{example_IC} with $\nu=1$ illustrating the distribution of eigenvalues. The limiting curve (solid black) is Matsuno
s density $\rho_\mathrm{M}(\lambda)$ normalized to unit mass.
} \label{fig:eigs_hist}
\end{center}
\end{figure}

In Figure~\ref{fig:phase_consts_asymp} we illustrate Theorem~\ref{theorem-gamma-limit} by plotting the real parts of the three exact phase constants $\gamma_\epsilon$ corresponding to eigenvalues $\lambda_\epsilon$ closest to $\Lambda=-1,-4,-5$ in the bulk $\mathcal{B}$.  The limiting values predicted by Theorem~\ref{theorem-gamma-limit} are indicated with dashed lines.  
The exact values of $\R\{\gamma_{\epsilon}\}$ were computed using the alternate formula for $\gamma_{\epsilon}$ presented in \cite{MillerWetzel2015}.

\begin{figure}[h!]
\begin{center}
\includegraphics[scale=1]{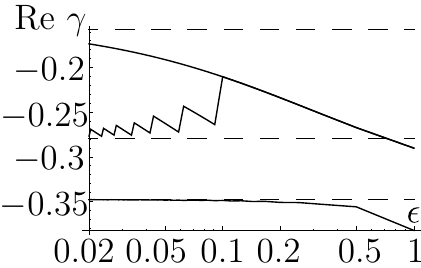}
\caption{
Three values of $\R\{\gamma_{\epsilon}\}$ (solid lines) as a function of $\epsilon$ for the rKS potential \eqref{example_IC} with $\nu = 1$. The limiting values (dashed lines) are from \eqref{equation-gamma-limit} evaluated for $\Lambda = -1,-4,-5$. The values $\R\{\gamma_{\epsilon}\}$ are computed at each $\epsilon$ with the exact eigenvalue $\lambda_{\epsilon}$ chosen to be the closest to $\Lambda$.
} \label{fig:phase_consts_asymp}
\end{center}
\end{figure}

The reflection coefficient is calculated exactly from \eqref{base_reflect_coeff} with the coefficients $v_p(\lambda)$ computed by solving the linear system \eqref{lin_sys_Abar}  for $\nu = \pm 1$. 
To aid in numerical computation we deform the path of integration in \eqref{base_reflect_coeff} into the complex plane to exploit the exponential decay of the integrand. 
We compare the exact reflection coefficient $\beta$ with its asymptotic approximation (Theorem~\ref{Reflect_asymptot}) in Figures~\ref{fig:beta_mag_neg} and \ref{fig:beta_phase_pos}. In Figure~\ref{fig:beta_mag_neg} we plot the normalized magnitude $\sqrt{\epsilon} |\beta|$ as a function of the spectral parameter $\lambda$.
The left panel ($\nu=-1$) shows that as $\epsilon\to 0$, $\sqrt{\epsilon}|\beta|$ indeed approaches the expected limit --- whose support is the bulk --- given by Matsuno's modulus formula (Corollary~\ref{beta_mag_asymp}).  The right panel ($\nu = 1$) shows that $\sqrt{\epsilon}|\beta|\to 0$ for positive rKS potentials as predicted by Theorem~\ref{Reflect_asymptot}.
In Figure~\ref{fig:beta_phase_pos} we plot the derivative of the phase of $\beta$ (computed indirectly using $\epsilon\I\{\beta'(\lambda)/\beta(\lambda)\}$) and compare with the corresponding limiting curve predicted by Theorem~\ref{Reflect_asymptot}, namely $-\theta_{+}'(\lambda)=-x_{+}(\lambda)$.

\begin{figure}[h!]
\begin{center}
\includegraphics[scale=1]{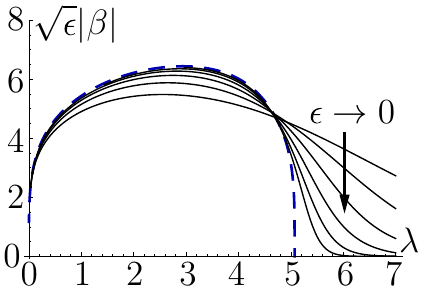}
\includegraphics[scale=1]{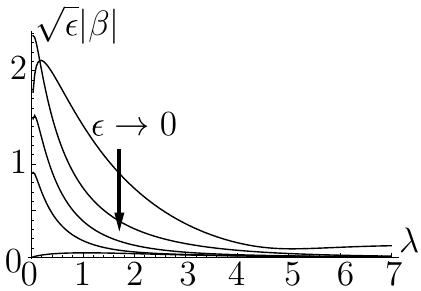}
\caption{Left: $\sqrt{\epsilon}|\beta|$ as a function of $\lambda$ for the rKS potential \eqref{example_IC} with $\nu = -1$ for $\epsilon = 2$, $1$, $1/2$, $1/4$, $1/8$, $1/16$ (solid black curves). The apparent limiting curve (dashed-blue) is obtained from Corollary~\ref{beta_mag_asymp}.  For $\epsilon$ and $\lambda$ both small the graphs become difficult to compute and are not plotted. 
Right: Same as the left panel but for $\nu = 1$ and $\epsilon = 4$, $2$, $7/4$, $13/8$, $3/2$, showing convergence to zero.
} \label{fig:beta_mag_neg}
\end{center}
\end{figure}

\begin{figure}[h!]
\begin{center}
\includegraphics[scale=1]{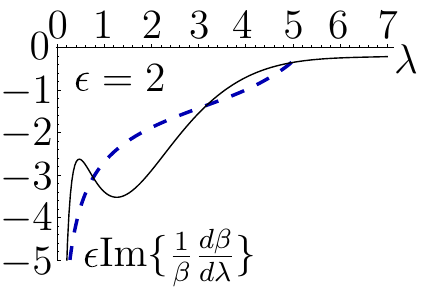}
\includegraphics[scale=1]{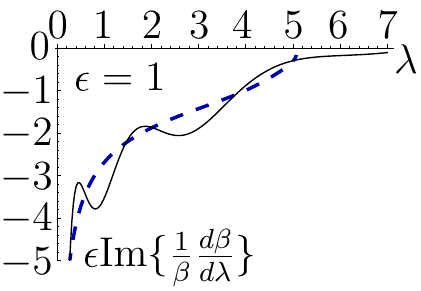}\\
\includegraphics[scale=1]{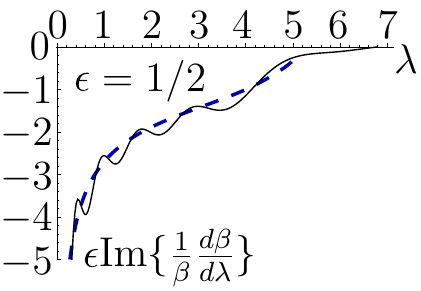}
\includegraphics[scale=1]{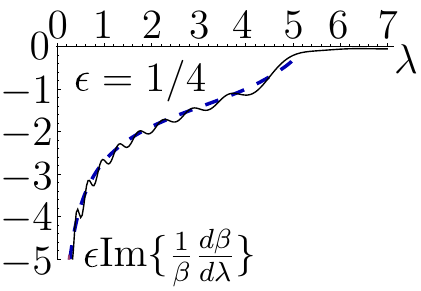}
\caption{Plots of the rescaled derivative of the phase for the function $\beta$ as a function $\lambda$. The limiting curve $-x_{+}(\lambda)$ (dashed-blue) is obtained from taking the derivative of the leading order phase $-\theta_{+}(\lambda)$ in \eqref{beta_asymp_form}.
} \label{fig:beta_phase_pos}
\end{center}
\end{figure}

\section*{Acknowledgments}
The authors were supported by the National Science Foundation under grant DMS-1206131 and wish to thank Sheehan Olver and Martin Strauss for useful conversations.

 \appendix
 \section*{Appendices}
 \renewcommand*{\thesection}{\Alph{section}}

\section{Proof of Proposition~\ref{prop-Reflection-contours}}\label{Appendix-Reflection}

Since $\I\{c_p\}>0$ for each $p$, the integrands for the elements of the matrix $\mathbf{A}(\lambda)$ and the vector $\mathbf{b}(\lambda)$ defined for $\lambda<0$ by \eqref{lin_sys_coeff} are integrable at all of the singular points $\{z_p\}_{p=1}^P$ in $\mathbb{C}^+$, so when $\lambda<0$, in each row for which the contour is $C_m=U_m^<$ (the general case) we may write the integrals in \eqref{lin_sys_coeff} in the form
\begin{equation}
\int_{C_m}e^{-ih(z;\lambda)/\epsilon}g(z)\,dz=\sum_{n=1}^{m}(1-e^{2\pi c_{n}/\epsilon})\int_{\ell_{n-1}(z_n)}e^{-ih(z;\lambda)/\epsilon}g(z)\,dz
\label{eq:collapsed-contours}
\end{equation}
where either $g(z)=(z-z_p)^{-1}$ or $g(z)=1$; see Figure~\ref{fig:U-combo-PhysicaD}.  Note that the term with $n=m$ in the sum in \eqref{eq:collapsed-contours} necessarily has a nonzero coefficient because otherwise $C_m=\ell_0(z_m)$ as $ic_m/\epsilon$ is a negative integer (the special case).  
Carrying out the analytic continuation of $\mathbf{A}(\lambda)$ and $\mathbf{b}(\lambda)$ through $\mathbb{C}^+$ and taking the boundary value on $\mathbb{R}^+$ we arrive at formulae for the matrix elements of $\mathbf{A}^>(\lambda)$ and $\mathbf{b}^>(\lambda)$ for $\lambda>0$, in which the infinite tails of all of the contours $\ell_{n-1}(z_n)$ for $n=1,\dots,m$ or $\ell_0(z_m)$ are all rotated through the left-half $z$-plane so that they now originate at $-i\infty$ to the left of the line $\R\{z\}=\R\{z_1\}$, and at the same time the integrand is analytically continued as necessary.
We define an intermediate set of contours $\{Y_m\}_{m=1}^P$ such that $Y_1$ connects $-i\infty$ (left of $\R\{z\}=\R\{z_1\}$) to $z_1$ and then for $m=2,\dots,P$, $Y_m$ connects $z_{m-1}$ to $z_m$, with $Y_m$ in the domain of $f(z)$.  Then the rotation of the integral over $\ell_{n-1}(z_n)$ in \eqref{eq:collapsed-contours} may be written as a linear combination of integrals with the same integrand over contours $Y_1,\dots,Y_n$, with the coefficient of the integral over $Y_n$ being $1$.  The same holds for the integral over $\ell_0(z_m)$ that arises in the special case.  The other coefficients in these combinations are nonzero factors arising from analytic continuation and the multi-valuedness of $f(z)$.  It is then clear that there is a lower triangular matrix $\mathbf{M}$ with nonzero diagonal entries $M_{mm}=(1-e^{2\pi c_m/\epsilon})$ in the general case and $M_{mm}=1$ in the special case such that for $m=1,\dots,P$, for $\lambda>0$ we have
\begin{equation}
\int_{\tilde{C}_m}e^{-ih(z;\lambda)/\epsilon}g(z)\,dz=\sum_{n=1}^m M_{mn}\int_{Y_n}e^{-ih(z;\lambda)/\epsilon}g(z)\,dz.
\end{equation}

Next we show how to exchange the contours $\{Y_m\}_{m=1}^P$ for a different set $\{W_m\}_{m=1}^P$ more amenable to asymptotic analysis and defined by the following procedure.
Consider the level curves of $R(z):=\R\{-ih(z;\lambda)\}=L$ for levels $L$ increasing from $L=-\infty$ to a sufficiently positive value.  Under the assumption that $\R\{c_p\}=0$ for all $p$, $R(z)$ is a harmonic function in the domain $\mathbb{C}\setminus\{z_1,\dots,z_P,z_1^*,\dots,z_P^*\}$.  The landscape of $R(z)$ can be visualized as being based upon a plane with slope $\lambda>0$ in the positive imaginary direction; this is the contribution of the term $\lambda z$ in $h(z;\lambda)$.  Superimposed on the plane are $P$ infinite mountain peaks at points $z_1^*,\dots,z_P^*$ in the lower half-plane and an abyss consisting of bottomless pits at the corresponding points $z_1,\dots,z_P$ in the upper half-plane.  
The latter are all features contributed from the term $f(z)$ in $h(z;\lambda)$.  If we visualize the level $L$ as the sea-level on the landscape of $R(z)$, then when $L$ is large and negative there will be small distinct lakes formed in each of the $P$ pits and otherwise the upper half-plane will be dry land.  Under a steady rain the sea level $L$ will rise, and at the critical values of $L=R(z)$ corresponding to the critical points of $h$ with positive imaginary parts, some lakes will fuse; the number of lakes thus decreases in steps as $L$ increases.  Meanwhile, the tide steadily rises from the lower half-plane, surrounding the mountains and forming islands that become isolated from each other at critical values corresponding to the critical points of $h$ with negative imaginary parts.  Eventually the lakes in the upper half-plane will begin to merge with the ocean as well as with each other, and ultimately (for large enough $L$) the shoreline in the upper half-plane will consist of a single curve stretching from horizon $(\R\{z\}=-\infty)$ to horizon ($\R\{z\}=+\infty$) as all lakes will have been subsumed by the rising tide.  Note that when $L=0$ the entire real axis is a component of the level curve.

Since the values of $R(z)$ at its critical points in $\mathbb{C}^+$ are distinct and nonzero, as the sea-level $L$ increases through a critical value, exactly two lakes will fuse, or a single lake will fuse with the ocean, with the point of fusion in each case being a critical point of $R$.  If $L=0$ is a critical value of $R$, then there are at most two real critical points $x_\pm(\lambda)$, distinct for $\lambda\in\mathcal{B}$, and at $L=0$ these are distinct fusion points where a single lake meets the ocean.   At the lowest critical level $L$, the fusing lakes each contain exactly one of the pits $\{z_p\}_{p=1}^P$ (or none, in the special case of fusion with the ocean), and we define a contour $W$ connecting the two pits (or connecting one pit with the point $-i\infty$ in the ocean) and lying below sea-level except at the critical point which it crosses locally in the steepest descent direction (if the lowest critical level is $L=0$ we make the connection via $x_-(\lambda)\in\mathbb{R}$).  Increasing $L$ then completely encloses $W$ within the new lake or ocean formed by the fusion.  Suppose that at some non-critical level $L$, each lake contains a chain of contours $W$ each of which connects a pair of pits (or for the ocean, exactly one of the links $W$ connects a pit with $-i\infty$) such that all enclosed pits are vertices of the chain.  This inductive hypothesis is certainly true for levels $L$ just above the lowest critical level $L$.  Increasing $L$ to the next higher critical level, define a new contour $W$ to pass over the fusion point locally in the steepest descent direction (choosing $x_-(\lambda)$ as the fusion point if $L=0$) but otherwise lying below the sea-level and connecting one of the finite endpoints of the chain (a pit) in each fusing lake or ocean.  Thus two chains are joined into a longer chain that is enclosed within the formed lake/ocean for slightly larger $L$, and the induction argument closes.  

Once all of the lakes have been subsumed by the ocean, we have assembled a single chain of contours $W$ with one finite endpoint (a pit), one infinite endpoint at $-i\infty$, and whose internal vertices are all of the remaining pits.  We now label the links of the chain in order from the infinite end to the finite end as $W_1,\dots,W_P$, and we orient each link consistent with this labeling.  Each contour $W_m$ contains exactly one critical point $z=x_m$ of $h(z;\lambda)$, the global maximizer of $R$ restricted to $W_m$, with $x_m\in\mathbb{C}^+$ unless both $m=1$ and $\lambda\in\mathcal{B}$, in which case $x_1=x_-(\lambda)\in\mathbb{R}$.  We concretely define the multi-valued function $h(z;\lambda)$ on $W_m$ by arbitrarily selecting a branch at $z=x_m$ and analytically continuing along $W_m$.  The last step is to observe that we may express integrals of $e^{-ih(z;\lambda)/\epsilon}g(z)$ along the contours $\{Y_m\}_{m=1}^P$ as linear combinations of integrals of the same integrand, with the branch chosen as indicated above, along the contours $\{W_m\}_{m=1}^P$, and vice-versa.  
But this is easy, because both chains of contours $\{W_m\}_{m=1}^P$ and $\{Y_m\}_{m=1}^P$ span the same set of nodes $\{z_p\}_{p=0}^P$, $z_0:=-i\infty$, so taking into account possible monodromy of the integrand, there indeed exists an invertible matrix $\mathbf{K}$ such that
\begin{equation}
\int_{W_n}e^{-ih(z;\lambda)/\epsilon}g(z)\,dz=\sum_{n=1}^P K_{mn}\int_{Y_n}e^{-ih(z;\lambda)/\epsilon}g(z)\,dz.
\end{equation}

Combining our results, we see that with $\mathbf{N}:=\mathbf{K}\mathbf{M}^{-1}$ the modified system $\mathbf{N}\mathbf{A}^>(\lambda)\mathbf{v}(\lambda)=\mathbf{N}\mathbf{b}^>(\lambda)$ having contours $\{W_m\}_{m=1}^P$ in place of $\{\tilde{C}_m\}_{m=1}^P$ is indeed suitable as claimed.

\section{Proof of Proposition~\ref{Wm_prop}}\label{Appendix-Evans}

We prove Proposition~\ref{Wm_prop} in the most interesting case that $\lambda\in\mathcal{B}$ and hence there exist two real critical points $x_\pm(\lambda)$, first associating a certain tree graph to the level sets of $R(z):=\R\{-ih(z;\lambda)\}$ generated by the
 potential $u_0$. The description of the level sets $R(z) = L$ is analogous to that presented in Appendix~\ref{Appendix-Reflection}, but the landscape is inverted, sloping downward as $\I\{z\}$ increases, and with ``pits'' and ``mountains'' interchanged.

\subsection{Construction of a Tree}
To construct a unique tree associated to the initial condition \eqref{u0_def_frac}, we consider the level sets $R(z) = L$ as $L$ is varied.
For large $L>0$, the level sets of $R$ in $\mathbb{C}^+$ will consist of $P$ disjoint loops each bounding a single ``island'' containing exactly one of the poles $\{z_p\}_{p=1}^P$, each of which is an infinite mountain peak 
since $i c_p/\epsilon >0$, and each of which is associated with a distinct leaf of the tree.  The $P$ islands will be associated with edges in the tree terminating at the leaves.
As $L$ is decreased to the first (largest) critical value, two islands will fuse 
at a critical point of the function $R(z)$ to form a single island containing two of the points $\{z_p\}_{p=1}^P$ as $L$ is decreases slightly further.
We associate these fusion points of islands with tree nodes having two outward-directed edges representing the two fusing islands and one inward-directed edge representing the newly-formed larger island.  
We include an additional node with the label $x_{\pm}$ for the island that merges onto the real line\footnote{Since the potential is assumed to be Klaus-Shaw, the level sets (islands) merge onto the real line only at the points $x_{\pm}$.}  (part of the level set $L=0$); for $L<0$ a ``continent'' invades $\mathbb{C}^+$ that may fuse with islands if there are any negative critical values of $R$. 
For convenience, we include an additional edge in the tree connecting the node associated with the lowest critical value $L$ to a root node labeled ``$-\infty$'' at height $L=-\infty$. 
Finally, we label the leaves of the tree  $\ell_1, \ldots, \ell_P$ in order of decreasing ``height'' $L$ of their parent nodes (breaking ties arbitrarily) and we label each node of the tree $x_p$ to correspond with its associated critical point.
The labeling of the leaves induces a permutation matrix $\mathbf{P}$ mapping the indices of the poles $\{z_p\}_{p=1}^P$ to those of the corresponding leaves.
We have thus constructed a rooted tree with $P$ internal nodes,
$P$ leaves $\{\ell_p\}_{p=1}^{P}$ at height $L = \infty$, and a root node $-\infty$ at height $L=-\infty$; see tree $A$ in Figure~\ref{fig:tree_pruning} for an example.

\subsection{Tree Pruning Algorithm}
We obtain suitable contours $\{W_m\}_{m=1}^P$ from the tree graph by using a ``pruning'' algorithm to identify all but one of the leaves of the tree (associated to poles $\{z_p\}_{p=1}^P$ via $\mathbf{P}$) with corresponding nodes (associated to the critical points $\{x_p\}_{p=1}^{P-1}$). 
In a manner to be explained, it is thus possible to associate a critical point and hence a level set which will be partly enclosed by the contour $W_m$; see the bottom panels of Figures~\ref{fig:ValleyHillPlot} and \ref{fig:ValleyHillPlotA}.

The first step is to chop (divide) the tree at the node $x_{\pm}$, with the split node $x_{\pm}$ becoming simultaneously a leaf of one subtree and the root of the other. We are thus left with two trees: one having $Q$ leaves, $Q-1$ internal nodes, and a root node labeled $x_{\pm}$, and another having $P-Q+1$ leaves (one of which is labeled $x_{\pm}$), $P-Q$ internal nodes, and a root node labeled $-\infty$.
The leaves of the subtree with root $x_\pm$ are necessarily those labeled $\ell_1,\dots,\ell_Q$ (all descendants of an internal node from a level $L>0$) while the leaves of the other subtree are $x_\pm$ and $\ell_{Q+1},\dots,\ell_P$.
The tree-pruning algorithm for the subtree with root $x_\pm$ is the following.

For $m = 1, \ldots, Q-1$
\begin{itemize}
\item Prune the leaf $\ell_m$ and remove its incident edge and parent node (joining the two remaining edges incident on the removed node).
\item Define $s_m$ as the formal sum of leaves $\ell_k$ already pruned from the subtree that were originally descendants of the removed node and associate $s_m$ to the removed node (and hence to a critical point $x_m$).
\end{itemize}

Note that for each $m=1,\dots,Q-1$ the sum $s_m$ contains $\ell_m$ and possibly some of the leaves $\ell_1,\dots,\ell_{m-1}$, but no leaves with indices exceeding $m$, and that only the leaf $\ell_Q$ does not appear in any sum.  We now repeat the algorithm on the subtree with root $-\infty$, with the index $m$ running over the range $Q+1,\dots,P$.  Again the sums $s_m$ will have a lower-triangular structure with $\ell_m$ appearing in the sum $s_m$, but the leaf $x_\pm$ will not appear in any sum.
After the algorithm is applied to both subtrees, we set $s_Q=\ell_1+\cdots + \ell_P$ and associate it with the node $x_{\pm}$.  

In this manner, the algorithm associates the formal sums $\{s_m\}_{m=1,\neq Q}^{P}$ uniquely to the complex critical points $\{x_{p}\}_{p=1}^{P-1}$.   Moreover, the coefficients in the formal sums $\{s_m\}_{m=1}^P$ form a $P\times P$ zero-one matrix $\mathbf{K}$ with block structure:
\begin{equation}
\mathbf{K}=\begin{bmatrix}\mathbf{L}_{Q-1} & \mathbf{0}_{(Q-1)\times (P-Q+1)}\\
\hdotsfor{2}\\
\mathbf{0}_{(P-Q)\times Q} & \mathbf{L}_{P-Q}\end{bmatrix}
\end{equation}
where $\mathbf{L}_{Q-1}$ and $\mathbf{L}_{P-Q}$ are square lower-triangular matrices of the indicated dimension with ones on the diagonal, and where the dots stand for a single row of $P$ ones:  $\mathbf{1}^{\mathsf{T}}$.  The matrix $\mathbf{K}$ is obviously invertible.

We now relate the formal sums $\{s_m\}_{m=1}^{P}$ to contours in the complex plane as follows. First consider $m\neq Q$, and let $x_m$ be the complex critical point associated with $s_m$; the level curve $R(z)=R(x_m)=L$ is therefore a ``figure-eight'' shoreline of two fusing islands (or one island fusing with the continent), and one lobe of the figure-eight contains exactly the poles $z_p$ corresponding to the leaves in the formal sum $s_m$.  We then take $W_m$ to be the contour consisting of (i) a short arc of the steepest descent path for $R(z)$ over the critical point $x_m$ and descending to the level $L-\delta$ for some small $\delta>0$, (ii) arcs of the level curve $R(z)=L-\delta$ meeting the arc from (i) and ending on one or two branch cuts of $f$ emanating from singularities associated to leaves in the formal sum $s_m$, and (iii) two vertical rays along those branch cuts going to $i\infty$.  
In the case that the figure-eight has an unbounded lobe (the continent), the algorithm guarantees that the bounded lobe will always be selected.  
So, the contour $W_m$ thus contains a unique critical point of $h$, which is the maximizer of $R$ on the contour.  The exceptional contour $W_Q$ is then related to the corresponding formal sum $s_Q=\ell_1,\dots,\ell_P$ by being defined to be homotopic to $U_P^<$ (thus enclosing \emph{all} of the singularities $\{z_p\}_{p=1}^P$) and to pass over the real critical points $x_\pm(\lambda)$ such that they act as simultaneous maximizers.

Finally, let $\tilde{U}_p$, $p=1,\dots,P$, be a U-shaped contour like those in Figure~\ref{fig:U-combo-PhysicaD} but enclosing only the point $z_p$ along with its branch cut.
The obvious relation $C_m=U_m^<=\tilde{U}_1+\cdots+\tilde{U}_m$ for $m=1,\dots,P$ clearly defines
a lower-triangular invertible zero-one matrix $\mathbf{M}$.
The contours $\{W_m\}_{m=1}^P$ therefore appear in the rows of the matrix $\tilde{\mathbf{A}}(\lambda)=\mathbf{N}\mathbf{A}(\lambda)$, where $\mathbf{N}:=\mathbf{K}\mathbf{P}\mathbf{M}^{-1}$ is invertible, and as each contour $W_m$ contains exactly one dominating critical point, with the exception of $W_Q$ which contains two both at level $L=0$, the proof is complete.

\begin{figure}[h!]
\begin{center}
\includegraphics[scale=1]{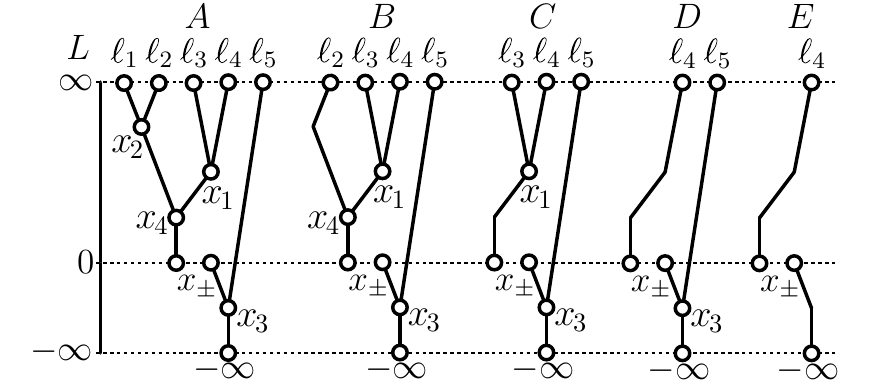}
\caption{
Example tree ($A$) and the tree pruning algorithm ($B$--$E$). Note that tree $A$ has been chopped at the node $x_{\pm}$ which is ``shared'' by the subtrees. The pruning algorithm on the subtree with root $x_\pm$ gives rise to the following assignments:
$x_2\mapsto (s_1 = \ell_1)$, $x_4\mapsto (s_2 = \ell_1 + \ell_2)$, and $x_1\mapsto (s_3 =  \ell_3)$.  Similarly, the pruning algorithm on the subtree with root $-\infty$ gives rise to the assignment $x_3\mapsto (s_5 = \ell_5)$. Finally, $x_\pm\mapsto (s_4 = \ell_1 + \ell_2 + \ell_3 + \ell_4 + \ell_5)$.
 } \label{fig:tree_pruning}
\end{center}
\end{figure}

\subsection{Elementary Examples}
Consider the positive rKS potential $u_0$ with $P=3$ and data
\begin{equation}\label{u0_sample1}
\begin{split}
\left(z_1,z_2,z_3\right)&=\tfrac{1}{2}\left(2i,2+4i,-2+i\right)\\
\left(c_1,c_2,c_3\right) &= -\tfrac{1}{2}i\left(2,1,4\right),
\end{split}
\end{equation}
and let $\lambda = -1$.  While it is not obvious from the data \eqref{u0_sample1}, from a graph of $u_0$ one can see that it satisfies the Klaus-Shaw condition.
The scheme for selection of appropriate contours $\{W_m\}_{m=1}^3$ in this case is illustrated in
Figure~\ref{fig:ValleyHillPlot}.
Obviously $W_j$ is homotopic on the domain of analyticity of $f$ to $C_j=U_j^<$ (see Figure~\ref{fig:U-combo-PhysicaD}), so here $\mathbf{N}=\mathbb{I}$ and hence $\tilde{\mathbf{A}}(\lambda)=\mathbf{A}(\lambda)$ and $\tilde{D}(\lambda)=D(\lambda)$.  Note that as desired, only $W_{P=3}$ traverses more than one critical point at the same level of $\R\{-ih(z;-1)\}$.

\begin{figure}[h!]
\begin{center}
\includegraphics[scale=1]{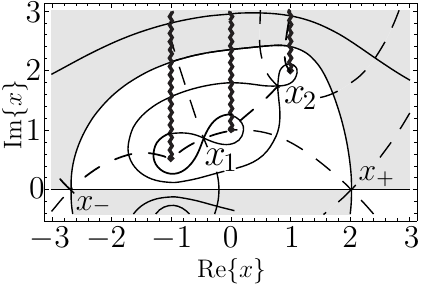}
\includegraphics[scale=1]{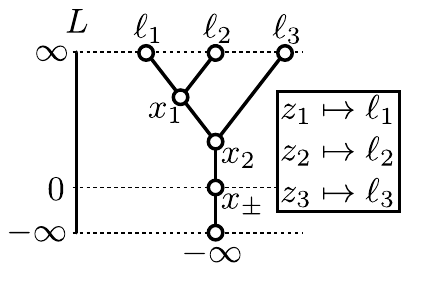}\\
\includegraphics[scale=1]{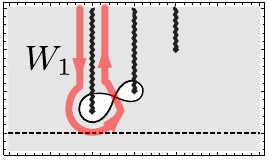}
\includegraphics[scale=1]{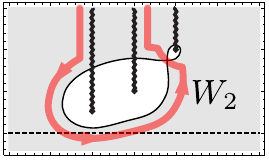}
\includegraphics[scale=1]{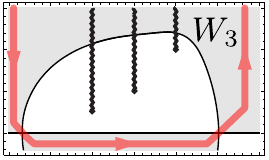}
\caption{
Top left panel: for the $P=3$ positive rKS potential $u_0$ with data given by \eqref{u0_sample1}, branch cuts of $f$ emanating from the branch points $\{z_1,z_2,z_3\}$ are shown with zigzagged lines, the solid curves are the levels of $\R\{-ih(x;-1)\}$, and the dashed curves are the levels of $\I\{-ih(x;-1)\}$ (steepest descent/ascent). The intersection points mark the critical points of $h(x;-1)$ which are numbered according to the values $\R\{-ih(x_j;-1)\}$ except for those ($x_\pm$) on the real axis. The domain $\R\{-ih(z;-1)\} < 0$ is shaded.
Top right panel: the tree associated with $h(x;-1)$, with the permutation $\mathbf{P}$ shown in the inset. 
Bottom panels: the contours $W_m\equiv U_m^<=C_m$ for which the integrals $\tilde{A}_{mp}(-1)$ are exponentially dominated by a contribution from neighborhoods of the critical point(s) over which the contour passes.  The level curve $\R\{-ih(z;-1)\}=\const$ containing the traversed critical point(s) is plotted, and the domain $\R\{-ih(z;-1)\}<\const$ is shaded in each case.
} \label{fig:ValleyHillPlot}
\end{center}
\end{figure}

To see that such a trivial outcome is not generally the case, we consider now replacing \eqref{u0_sample1} with 
\begin{equation}
\label{u0_sample2}
\begin{split}
\left(z_1,z_2,z_3\right)&=\tfrac{1}{2}\left(-4+6i,2i,2+i\right)\\
\left(c_1,c_2,c_3\right) &= -\tfrac{1}{3}i\left(3,1,3\right),
\end{split}
\end{equation}
again yielding a positive rKS potential with $P=3$.  See Figure~\ref{fig:ValleyHillPlotA}.  The matrix $\mathbf{N}$ returned by the algorithm is
\begin{equation}
\mathbf{K}\mathbf{P}\mathbf{M}^{-1}=
\begin{bmatrix}1 & 0 & 0\\1 & 1 & 1\\0 & 0 & 1\end{bmatrix}
\begin{bmatrix}0 & 0 & 1\\0 & 1 & 0\\1 & 0 & 0\end{bmatrix}
\begin{bmatrix}1 & 0 & 0\\1 & 1 & 0\\1 & 1 & 1\end{bmatrix}^{-1}=
\begin{bmatrix}
0& -1 & 1\\
0 & 0 & 1\\
1 & 0 & 0
\end{bmatrix}.
\end{equation}
\begin{figure}[h!]
\begin{center}
\includegraphics[scale=1]{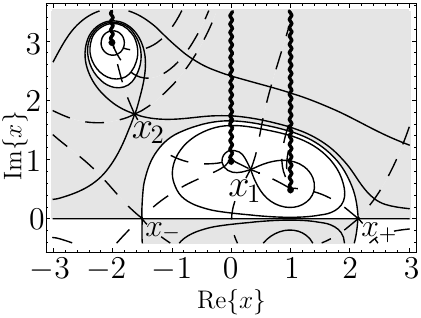}
\includegraphics[scale=1]{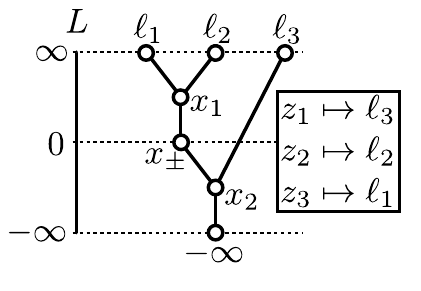}\\
\includegraphics[scale=1]{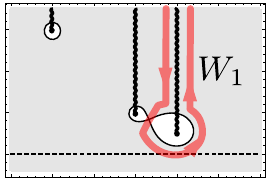}
\includegraphics[scale=1]{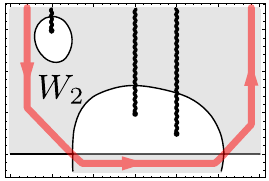}
\includegraphics[scale=1]{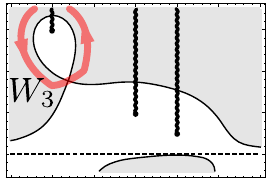}
\caption{Same as Figure~\ref{fig:ValleyHillPlot} but for the positive rKS potential with \eqref{u0_sample2}. } \label{fig:ValleyHillPlotA}
\end{center}
\end{figure}

\section{Proof of Proposition~\ref{proposition-I-epsilon}}\label{Appendix-Phase}

Recall the vector $\boldsymbol{\phi}(\lambda)$ characterized for an eigenvalue $\lambda<0$ by the
rank $P-1$ homogeneous system $\mathbf{A}(\lambda)\boldsymbol{\phi}(\lambda)=\mathbf{0}$ and the normalization condition $\mathbf{1}^\mathsf{T}\boldsymbol{\phi}(\lambda)=\lambda$, where $\mathbf{1}^\mathsf{T}:=(1,1,\dots,1)$.  Recall also the eigenvalue $\lambda_\epsilon$ nearest $\Lambda\in\mathcal{B}$.
\begin{lem}\label{Phi_eps_bound}
For each $\Lambda\in\mathcal{B}$, the elements $\phi_p(\lambda_\epsilon)$, $p=1,\dots,P$ of the vector $\boldsymbol{\phi}(\lambda_\epsilon)$ satisfy
\begin{equation}\label{Phij0_def}
\lim_{\epsilon \to 0}\phi_{p}(\lambda_\epsilon)=\phi_p^0:=
 \Lambda \underset{z=z_p}{\Res}\frac{\Psi^{-}(z;\Lambda)}{z-x_{-}(\Lambda)}
= \Lambda \frac{ \prod_{l =1}^{P-1} \left(z_p-x_l(\Lambda) \right)}{ \prod_{l \neq p} \left(z_p - z_l\right)},
\end{equation}
where $\Psi^{-}$ is defined by \eqref{Psi_defn}.
\end{lem}

\begin{proof}
By taking appropriate linear combinations of the rows and multiplying on the left by a suitable invertible diagonal matrix $\mathbf{D}(\lambda_\epsilon)$ to be determined, we replace the system $\mathbf{A}(\lambda_\epsilon)\boldsymbol{\phi}(\lambda_\epsilon)=\mathbf{0}$ by the system $\mathbf{D}(\lambda_\epsilon)\tilde{\mathbf{A}}(\lambda_\epsilon)\boldsymbol{\phi}(\lambda_\epsilon)=\mathbf{0}$, maintaining the rank as $P-1$.  Next,
we replace the (redundant) $P^\text{th}$ row of $\mathbf{D}(\lambda_\epsilon)\tilde{\mathbf{A}}(\lambda_\epsilon)$ with $\mathbf{1}^\mathsf{T}$, resulting in a matrix we denote as $\hat{\mathbf{A}}(\lambda_\epsilon)$.  Since $\mathbf{1}^\mathsf{T}\boldsymbol{\phi}(\lambda_\epsilon)=\lambda_\epsilon$,  $\boldsymbol{\phi}(\lambda_\epsilon)$ is
a solution of the square inhomogeneous linear system $\hat{\mathbf{A}}(\lambda_\epsilon)\boldsymbol{\phi}(\lambda_\epsilon)=\hat{\mathbf{b}}(\lambda_\epsilon):=(0,0,\dots,0,\lambda_\epsilon)^\mathsf{T}$.  We prove the lemma by showing that for small enough $\epsilon$ this system has full rank $P$ and by explicitly constructing the unique limiting solution.  

Recalling \eqref{Dm_asymp} and Remark~\ref{remark-steepest-descent-lambda-perturb}, we see that upon choosing 
\begin{equation}
D_m(\lambda_\epsilon):=e^{ih(x_m(\Lambda);\lambda_\epsilon)/\epsilon - i\theta_m(\Lambda)}
\sqrt{\frac{|h''(x_m(\Lambda);\Lambda)|}{2\pi\epsilon}}\neq 0 
\end{equation}
for $m=1,\dots,P-1$, we obtain $\lim_{\epsilon\to 0}\hat{\mathbf{A}}(\lambda_\epsilon)=\hat{\mathbf{A}}^0$, where
\begin{equation}
\hat{A}^0_{mp}=\frac{1}{x_m(\Lambda)-z_p},\quad m=1,\dots,P-1,\quad p=1,\dots,P
\end{equation}
while the last row of $\hat{\mathbf{A}}^0$ is $\mathbf{1}^\mathsf{T}$ (unchanged in the limit).  Also $\lim_{\epsilon\to 0}\hat{\mathbf{b}}(\lambda_\epsilon)=\hat{\mathbf{b}}^0:=(0,0,\dots,0,\Lambda)^\mathsf{T}$.  By a direct computation one sees that $\det(\hat{\mathbf{A}}^0)\neq 0$ because the critical points $x_1(\Lambda),\dots,x_{P-1}(\Lambda)$ are simple and hence distinct, as are the poles $z_1,\dots,z_P$.  

It therefore only remains to show that the unique solution of the limiting system $\hat{\mathbf{A}}^0\boldsymbol{\phi}^0=\hat{\mathbf{b}}^0$ is given by \eqref{Phij0_def}, a fact that is a simple consequence of the Residue Theorem; see \cite{Schechter59} for a complementary approach.
Indeed, for $m=1,\dots,P-1$, 
\begin{equation}
\begin{split}
\sum_{p=1}^{P} \hat{A}^0_{mp}\phi^0_{p}
&= - \Lambda \sum_{p=1}^{P}  \frac{1}{z_p - x_m(\Lambda) } \underset{z=z_p}{\Res}\frac{\Psi^{-}(z;\Lambda)}{z-x_{-}(\Lambda)}\\
&= -\frac{\Lambda}{2 \pi i} \oint_C \frac{\Psi^{-}(z;\Lambda)}{\left(z-x_m(\Lambda) \right) \left(z-x_{-}(\Lambda)\right)} \, dz 
\end{split}
\end{equation}
where $C$ is a counter-clockwise contour about the poles of the integrand $\{z_p\}_{p=1}^{P}$. On the other hand, this integral vanishes because the integrand is analytic in the exterior of $C$ and is $O(z^{-2})$ as $z\to\infty$.
Similarly, the last row of the product $\hat{\mathbf{A}}^0\boldsymbol{\phi}^0$ is
\begin{equation}
\sum_{p=1}^{P}\phi^0_{p}
= \Lambda \sum_{p=1}^{P} \underset{z=z_p}{\Res}\frac{\Psi^{-}(z;\Lambda)}{z-x_{-}(\Lambda)}
= \frac{\Lambda}{2 \pi i} \oint_C \frac{\Psi^{-}(z;\Lambda)}{z-x_{-}(\Lambda)} \; dz = \Lambda,
\end{equation}
the integral being computed by a residue at $z=\infty$.
\end{proof}

Recall that $\Phi_\epsilon(x)=\Phi(x;\lambda_\epsilon)$ is given by \eqref{Phij_int} with $\lambda=\lambda_\epsilon$.  Also, since $\lambda_\epsilon$ is an eigenvalue, $\Phi_\epsilon(x)\to 0$ as $x\to+\infty$, which implies that the lower limit of integration in \eqref{Phij_int} can be replaced by $+\infty$ when $\lambda=\lambda_\epsilon$.  Using the definition of $I_\epsilon(x)$ given in \eqref{gamma_eps} then shows that  (the two signs $\pm$ are arbitrary but equal)
\begin{equation}
I_\epsilon(x)=\frac{1}{\lambda_\epsilon}\sum_{p=1}^P\sum_{q=1}^P\phi_p(\lambda_\epsilon)\phi_q(\lambda_\epsilon)^*a^\pm_{p,\epsilon}(x)a^\pm_{q,\epsilon}(x)^*,
\label{eq:I-eps-rewrite}
\end{equation}
where
\begin{equation}
a_{p,\epsilon}^\pm(x):=\frac{1}{\sqrt{2\pi\epsilon}}\int_{\pm\infty}^x\frac{e^{-ih_\epsilon(z)/\epsilon}\,dz}{z-z_p},\quad h_\epsilon(z):=h(z;\lambda_\epsilon).
\label{eq:a-p-epsilon}
\end{equation}

We prove properties 1--3 of Proposition~\ref{proposition-I-epsilon} in turn.
\paragraph{Property 1}  To prove \eqref{I_eps_tail_bound}, suppose that $x\le -L$.  Integrating by parts using $(\lambda_\epsilon+u_0(z))e^{-ih_\epsilon(z)/\epsilon}=i\epsilon \partial _ze^{-ih_\epsilon(z)/\epsilon}$ gives
\begin{multline}
a^-_{p,\epsilon}(x)=i\sqrt{\frac{\epsilon}{2\pi}}\left[\frac{e^{-ih_\epsilon(x)/\epsilon}}{(\lambda_\epsilon+u_0(x))(x-z_p)}\right. \\{}+ \left.\int_{-\infty}^x\frac{e^{-ih_\epsilon(z)/\epsilon}u_0'(z)\,dz}{(\lambda_\epsilon+u_0(z))^2(z-z_p)} +\int_{-\infty}^x\frac{e^{-ih_\epsilon(z)/\epsilon}\,dz}{(\lambda_\epsilon+u_0(z))(z-z_p)^2}\right].
\label{eq:tail-ibp-1}
\end{multline}
To control the last term we integrate by parts again:
\begin{multline}
\int_{-\infty}^x\frac{e^{-ih_\epsilon(z)/\epsilon}\,dz}{(\lambda_\epsilon+u_0(z))(z-z_p)^2}=i\epsilon\left[
\frac{e^{-ih_\epsilon(x)/\epsilon}}{(\lambda_\epsilon+u_0(x))^2(x-z_p)^2}\right.\\
{}+\left.\int_{-\infty}^x\frac{2e^{-ih_\epsilon(z)/\epsilon}u_0'(z)\,dz}{(\lambda_\epsilon+u_0(z))^3(z-z_p)^2}
+\int_{-\infty}^x\frac{2e^{-ih_\epsilon(z)/\epsilon}\,dz}{(\lambda_\epsilon+u_0(z))^2(z-z_p)^3}\right].
\label{eq:tail-ibp-2}
\end{multline}
The integrals in \eqref{eq:tail-ibp-1}--\eqref{eq:tail-ibp-2} are, unlike that on the right-hand side of \eqref{eq:a-p-epsilon}, absolutely convergent.  Since $\lambda_\epsilon+u_0(x)$ is bounded away from zero uniformly for $\epsilon$ sufficiently small and $x\le -L$, and since there exists $K>0$ such that $|u_0'(x)|\le K|x|^{-3}$ holds for all $x\in\mathbb{R}$, the dominant term in $a^-_{p,\epsilon}(x)$ comes from the first term on the right-hand side of \eqref{eq:tail-ibp-1}, and we easily obtain
\begin{equation}
\left| a^-_{p,\epsilon}(x)-i\sqrt{\frac{\epsilon}{2\pi}}\frac{e^{-ih_\epsilon(x)/\epsilon}}{\lambda_\epsilon x}\right|\le\frac{K'\sqrt{\epsilon}}{x^2},\;\; x\le -L,\;\; p=1,\dots,P.
\label{eq:a-minus-tail-estimate}
\end{equation}
for some constant $K'>0$.  Using \eqref{eq:a-minus-tail-estimate} in \eqref{eq:I-eps-rewrite} and taking into account $\mathbf{1}^\mathsf{T}\boldsymbol{\phi}(\lambda_\epsilon)=\lambda_\epsilon$ as well as the boundedness of $\boldsymbol{\phi}(\lambda_\epsilon)$ in the limit $\epsilon\to 0$ implied by Lemma~\ref{Phi_eps_bound} proves \eqref{I_eps_tail_bound} for $x\le -L$.  For $x\ge L$, one simply uses \eqref{eq:I-eps-rewrite} with $a^+_{p,\epsilon}(x)$ in place of $a^-_{p,\epsilon}(x)$.

\paragraph{Property 2}  To prove \eqref{I_eps_uniform_bound}, it follows from \eqref{eq:I-eps-rewrite} and Lemma~\ref{Phi_eps_bound} that it is sufficient to uniformly bound $a^-_{p,\epsilon}(x)$ for $x\le\overline{x}(\Lambda)$ and $a^+_{p,\epsilon}(x)$ for $x\ge\overline{x}(\Lambda)$, because $\lambda_\epsilon\to\Lambda<0$ as $\epsilon\to 0$.  We will prove that $|a_{p,\epsilon}^-(x)|\le K$ holds on $x\le\overline{x}(\Lambda)$ for all $p=1,\dots,P$ and all $\epsilon>0$ sufficiently small, with the complementary estimate of $|a^+_{p,\epsilon}(x)|$ on $x\ge\overline{x}(\Lambda)$ being similar.  Moreover, from \eqref{eq:a-minus-tail-estimate} we see that $|a^-_{p,\epsilon}(x)|\le K$ certainly holds for $x\le -L$, so it is enough to consider the difference $a^-_{p,\epsilon}(x)-a^-_{p,\epsilon}(-L)$ for $-L\le x\le \overline{x}(\Lambda)$.  

The limiting exponent function $h(z;\Lambda)$ has a unique critical point at $z=x_-(\Lambda)$ in the interval $(-L,\overline{x}(\Lambda))$ at which $h''(x_-(\Lambda);\Lambda)=u_0'(x_-(\Lambda))>0$ because $u_0$ is positive Klaus-Shaw.  The equation $h(z;\Lambda)-h(x_-(\Lambda);\Lambda)=s^2$ therefore defines a unique mapping $z=g(s)$ with $g'(0)=\sqrt{2/u_0'(x_-(\Lambda))}>0$ that is conformal on a complex neighborhood of $s=0$ containing a square $\max\{|\R\{s\}|,|\I\{s\}|\}\le\delta'$
whose conformal image in turn contains a neighborhood of the form $|z-x_-(\Lambda)|<\delta$ for some $\delta=\delta(\Lambda)>0$ (independent of $\epsilon$).  If $-L\le x\le x_-(\Lambda)-\delta$, we integrate by parts using $(\lambda_\epsilon+u_0(z))e^{-ih_\epsilon(z)/\epsilon}=i\epsilon \partial_z e^{-ih_\epsilon(z)/\epsilon}$ to obtain
\begin{equation}
\left|a^-_{p,\epsilon}(x)-a_{p,\epsilon}^-(-L)\right|=
\frac{1}{\sqrt{2\pi\epsilon}}\left|\int_{-L}^x\frac{e^{-ih_\epsilon(z)/\epsilon}\,dz}{z-z_p}\right|\le K'\sqrt{\epsilon}
\label{eq:uniform-bound-left}
\end{equation}
for some $K'=K'(\delta)$ independent of both $\epsilon$ and $x\in [-L,x_-(\Lambda)-\delta]$.  If
instead $x_-(\Lambda)+\delta\le x\le \overline{x}(\Lambda)$, we first integrate from $-L$ to $\overline{x}(\Lambda)$ and then back to $x$, and obtain by the same argument
\begin{equation}\label{eq:uniform-bound-right}
\left|a^-_{p,\epsilon}(x)-a^-_{p,\epsilon}(-L)-\frac{1}{\sqrt{2\pi\epsilon}}\int_{-L}^{\overline{x}(\Lambda)}
\frac{e^{-ih_\epsilon(z)/\epsilon}\,dz}{z-z_p}\right| = \frac{1}{\sqrt{2\pi\epsilon}}\left|\int_x^{\overline{x}(\Lambda)}\frac{e^{-ih_\epsilon(z)/\epsilon}\,dz}{z-z_p}\right| \le K'\sqrt{\epsilon}.
\end{equation}
The integral on the left-hand side (over $-L\le z\le\overline{x}(\Lambda)$) is independent of $x$, and by the method of stationary phase (appropriately generalized as described in Remark~\ref{remark-steepest-descent-lambda-perturb} using $|\lambda_\epsilon-\Lambda|\le C\epsilon$) it is $O(\sqrt{\epsilon})$ due to the simple critical point at $z=x_-(\Lambda)$.  Therefore, combining \eqref{eq:uniform-bound-left}--\eqref{eq:uniform-bound-right} shows that $|a^-_{p,\epsilon}(x)|\le K$ also holds for $-L\le x\le x_-(\Lambda)-\delta$ and for $x_-(\Lambda)+\delta\le x\le \overline{x}(\Lambda)$.
Finally, for $|x-x_-(\Lambda)|<\delta$ we write
\begin{equation}\label{eq:a-minus-diff-center}
a^-_{p,\epsilon}(x)-a^-_{p,\epsilon}(-L)=\frac{1}{\sqrt{2\pi\epsilon}}\int_{-L}^{x_-(\Lambda)}
\frac{e^{-ih_\epsilon(z)/\epsilon}\,dz}{z-z_p} +\frac{1}{\sqrt{2\pi\epsilon}} \int_0^{g^{-1}(x)}e^{-is^2/\epsilon}k_{p,\epsilon}(s)\,ds
\end{equation}
where
\begin{equation}
k_{p,\epsilon}(s):=\frac{e^{i(\Lambda-\lambda_\epsilon)g(s)/\epsilon}g'(s)}{g(s)-z_p}.
\end{equation}
The first term on the right-hand side of \eqref{eq:a-minus-diff-center} is independent of $x$ and is $O(1)$ as $\epsilon\to 0$ by (generalized) stationary phase.  To analyze the second term, we assume that $x>x_-(\Lambda)$ (the case $x<x_-(\Lambda)$ is similar) which implies that $g^{-1}(x)>0$, and apply Cauchy's Theorem
in the square $\max\{|\R\{s\}|,|\I\{s\}|\}\le\delta'$, on which $k_{p,\epsilon}(s)$ is analytic and bounded uniformly for small $\epsilon$ because $|\lambda_\epsilon-\Lambda|\le C\epsilon$, to deform the real path to a diagonal and a vertical segment:
\begin{equation}
\int_0^{g^{-1}(x)}e^{-is^2/\epsilon}k_{p,\epsilon}(s)\,ds = \int_0^{(1-i)g^{-1}(x)}e^{-is^2/\epsilon}k_{p,\epsilon}(s)\,ds +
\int_{(1-i)g^{-1}(x)}^{g^{-1}(x)}e^{-is^2/\epsilon}k_{p,\epsilon}(s)\,ds.
\label{eq:diag-vert}
\end{equation}
But because $|k_{p,\epsilon}(s)|\le K''$, 
\begin{equation}
\left|\int_0^{(1-i)g^{-1}(x)}e^{-is^2/\epsilon}k_{p,\epsilon}(s)\,ds\right|=
K''\int_0^{\sqrt{2}g^{-1}(x)}e^{-t^2/\epsilon}\,dt<\frac{K''\sqrt{\pi\epsilon}}{2}
\end{equation}
and similarly
\begin{multline}
\left|\int_{(1-i)g^{-1}(x)}^{g^{-1}(x)}e^{-is^2/\epsilon}k_{p,\epsilon}(s)\,ds\right|\le K''\int_0^{g^{-1}(x)}
e^{-g^{-1}(x)t/\epsilon}\,dt\\
{}\le K''\int_0^{g^{-1}(x)}e^{-t^2/\epsilon}\,ds<\frac{K''\sqrt{\pi\epsilon}}{2}.
\label{eq:vertical-integral-estimate}
\end{multline}
Using \eqref{eq:diag-vert}--\eqref{eq:vertical-integral-estimate} in \eqref{eq:a-minus-diff-center} then shows that $|a^-_{p,\epsilon}(x)|\le K$ holds also for $|x-x_-(\Lambda)|<\delta$.

\paragraph{Property 3}   Finally, the pointwise limit \eqref{I_eps_pointwise_limit} follows from an application of the method of stationary phase (or steepest descent) to $a_{p,\epsilon}^-(x)$ (for $x\le\overline{x}(\Lambda)$, $x\neq x_-(\Lambda)$) and $a_{p,\epsilon}^+(x)$ (for $x\ge\overline{x}(\Lambda)$, $x\neq x_+(\Lambda)$), taking into account Remark~\ref{remark-steepest-descent-lambda-perturb} and the fact that $|\lambda_\epsilon-\Lambda|\le C\epsilon$.  The integrals in \eqref{eq:a-p-epsilon} involve contributions from stationary phase points only if $x_-(\Lambda)<x<x_+(\Lambda)$, and hence $a^-_{p,\epsilon}(x)\to 0$ pointwise for $x<x_-(\Lambda)$ while $a^+_{p,\epsilon}(x)\to 0$ pointwise for $x>x_+(\Lambda)$, implying the same trivial limit for $I_\epsilon(x)$.  For $x\in (x_-(\Lambda),x_+(\Lambda))$, the integrals $a^-_{p,\epsilon}(x)$ have a unique simple stationary phase point at $z=x_-(\Lambda)$, with the result that
\begin{equation}
a^-_{p,\epsilon}(x)=\frac{e^{-i\pi/4}e^{-ih(x_-(\Lambda);\lambda_\epsilon)/\epsilon}}{(x_-(\Lambda)-z_p)\sqrt{u_0'(x_-(\Lambda))}}+O(\sqrt{\epsilon}),\quad\epsilon \to 0.
\end{equation}
Using this result in \eqref{eq:I-eps-rewrite} and using 
Lemma~\ref{Phi_eps_bound} then gives, for $x_-(\Lambda)<x<x_+(\Lambda)$,
\begin{equation}
\begin{split}
\lim_{\epsilon\to 0}I_\epsilon(x)&=\frac{1}{\Lambda u_0'(x_-(\Lambda))}\left|\sum_{p=1}^P\frac{\phi_p^0}{x_-(\Lambda)-z_p}\right|^2\\
&=\frac{\Lambda}{u_0'(x_-(\Lambda))}\left|\sum_{p=1}^P\frac{\prod_{l=1}^{P-1}(z_p-x_l(\Lambda))}{ (x_-(\Lambda)-z_p)\prod_{l\neq p}(z_p-z_l)}\right|^2\\
&=\frac{\Lambda}{u_0'(x_-(\Lambda))}\left|\frac{1}{2\pi i}\oint_C \frac{
\prod_{l=1}^{P-1}(z-x_l(\Lambda))}{(x_-(\Lambda)-z)\prod_{l=1}^p(z-z_l)}\,dz\right|^2,
\end{split}
\end{equation}
where $C$ is a positively-oriented contour enclosing $z_1,\dots,z_P$ but excluding $x_-(\Lambda)$.
Evaluating this integral by taking an exterior residue at $z=x_-(\Lambda)$ (no residue at $\infty$) and comparing with \eqref{Psi_defn} and \eqref{u_prime_poly} establishes the limit \eqref{I_eps_pointwise_limit} when $x_-(\Lambda)<x<x_+(\Lambda)$.

\end{document}